\documentclass[english]{amsart}
\usepackage{xr}
\usepackage{hyperref, caption}
\usepackage[english]{babel}
\usepackage{bbm}
\usepackage{hyperref}
\usepackage{verbatim}
\usepackage{ifthen}
\usepackage{amsfonts,amssymb,mathtools,amsthm,amsmath}
\usepackage{graphicx,xcolor,subcaption}
\usepackage{bbm}

\usepackage{adjustbox}

\usepackage{array, booktabs, multirow, hhline}
\renewcommand{\arraystretch}{1.5}

\makeatletter
\newcommand*{\centerfloat}{%
  \parindent \z@
  \leftskip \z@ \@plus 1fil \@minus \textwidth
  \rightskip\leftskip
  \parfillskip \z@skip}
\makeatother

\everymath{\displaystyle}

\usepackage[style=numeric,maxcitenames=2,maxbibnames=10,doi=false,isbn=false,url=false,natbib=true,giveninits=true,hyperref,bibencoding=utf8,backend=biber]{biblatex}
\AtEveryBibitem{%
  \clearfield{note}%
  \clearfield{issn}%
  \clearlist{language}%
  }

\newcommand{\R}{\ensuremath{\mathbb{R}}}

\newcommand{\E}{\ensuremath{\mathbb{E}}}
\renewcommand{\P}{\ensuremath{\mathbb{P}}}

\newcommand{\ind}[1]{\ensuremath{\mathbbm{1}_{\left\{#1\right\}}}}
\newcommand{\diff}{\mathop{}\mathopen{}\mathrm{d}}
\newcommand{\cal}[1]{\ensuremath{\mathcal{#1}}}

\newcommand\steq[1]{\stackrel{\text{\rm #1}.}{=}}

\def\eps{\varepsilon}

\newtheorem{proposition}{Proposition}
\newtheorem{definition}{Definition}
\newtheorem{lemma}{Lemma}
\newtheorem{theorem}{Theorem}

\title[Stochastic Models with Pair-Based STDP]{On the Spontaneous Dynamics of Synaptic Weights in Stochastic Models
with Pair-Based STDP}

\date{\today}

\author[Ph. Robert]{Philippe Robert}
\email{Philippe.Robert@inria.fr}
\urladdr{http://www-rocq.inria.fr/who/Philippe.Robert}
\address[Ph.~Robert, G.~Vignoud]{INRIA Paris, 2 rue Simone Iff, 75589 Paris Cedex 12, France}
\author[G. Vignoud]{Ga\"etan Vignoud${}^1$}
\email{Gaetan.Vignoud@inria.fr}
\address[G. Vignoud]{ Center for Interdisciplinary Research in Biology (CIRB) - Coll\`ege de France (CNRS UMR 7241, INSERM U1050), 11 Place Marcelin Berthelot, 75005 Paris, France}
\thanks{${}^1$Supported by PhD grant of \'Ecole Normale Sup\'erieure, ENS-PSL}

\ifpdf
\hypersetup{
  pdftitle={On the Spontaneous Dynamics of Synaptic Weights in Stochastic Models
with Pair-Based STDP},
  pdfauthor={Philippe Robert and Ga\"etan Vignoud}
}
\fi

\begin{document}

\begin{abstract}
We investigate spike-timing dependent plasticity (STPD) in the case of a synapse connecting two neural cells.
We develop a theoretical analysis of several STDP rules using Markovian theory.
In this context there are two different timescales, fast neural activity and slower synaptic weight updates.
Exploiting this timescale separation, we derive the long-time limits of a single synaptic weight subject to STDP.
We show that the pairing model of presynaptic and postsynaptic spikes controls the synaptic weight dynamics for small
external input, on  an excitatory synapse.
This result implies in particular that mean-field analysis of plasticity may miss some important properties of STDP.
Anti-Hebbian STDP seems to favor the emergence of a stable synaptic weight, but only for high external input.
In the case of inhibitory synapse the pairing schemes matter less, and we observe convergence of the synaptic weight to
a non-null value only for Hebbian STDP.
We extensively study different asymptotic regimes for STDP rules, raising interesting questions for future works on adaptative neural networks and, more generally, on adaptive systems.
\end{abstract}

\maketitle

\section{Introduction}
\label{sec:intro}
Understanding brain's learning and memory is a challenging topic combining a large spectrum of research fields ranging
from neurobiology to applied mathematics.
Neural networks, through the dynamics of their connections, are able to store complex patterns
over long periods of time, and as such are good candidates for the establishment of memory.
In particular, the intensity $W$ of the connection between two neurons, {\em the synaptic weight}, is seen as an
essential building block to explain learning and memory formation~\cite{takeuchi_synaptic_2014}.

Synaptic plasticity, processes that can modify the synaptic weight, is a complex
mechanism~\cite{citri_synaptic_2008}, but general principles have been inferred from experimental data and used for a
long time in computational models.

Spike-timing dependent plasticity (STDP) gathers plasticity processes that depends on the timing of pre-synaptic
and post-synaptic spiking activity.
Many experimental protocols has been developed to study STDP: most use sequences of spikes pairing from
either side of a specific synapse are presented, at a certain frequency and with a certain delay,
see~\cite{feldman_spike-timing_2012}.

Experiments show that long-term synaptic plasticity is characterized by the coexistence of two different timescales.
Membrane potential and pre/post-synaptic interspike intervals evolve on the order of several
milliseconds, see~\cite{gerstner_spiking_2002}.
Synaptic weights $W$ change on a slower timescale ranging from seconds to minutes before observing an effect of an STDP
protocol on the synaptic weights.
For this reason, a slow-fast approximation is proposed to analyze the associated mathematical models of synaptic
plasticity.
The analysis of slow-fast limits for a general class of STDP models is detailed in~\cite{robert_stochastic_2020_1,
robert_stochastic_2020_2,robert_stochastic_2020}.

Computational models of synaptic plasticity have also used similar scaling principles, see~\cite{kempter_hebbian_1999,
roberts_computational_1999, kistler_modeling_2000, van_rossum_stable_2000}.

In pair-based models,  the synaptic weight updates depend only on $\Delta t{=}t_{\rm post}-t_{\rm pre}$ for a subset of
instants of pre/post-synaptic spikes $t_{\rm pre}/t_{\rm post}$.

{\em Hebbian STDP} plasticity occur when
\begin{itemize}
\item a {\em pre-post pairing}, i.e., $t_{\rm pre}{<}t_{\rm post}$ leads to an increase of the synaptic weight
value (potentiation), which translates into $\Delta W{>}0$;
\item a {\em post-pre pairing}, $t_{\rm post}{<}t_{\rm pre}$, leads to a smaller synaptic weight (depression),
and therefore $\Delta W{<}0$.
\end{itemize}

Hebbian STDP has been observed at many different synapses~\cite{bi_synaptic_1998, feldman_spike-timing_2012}
and is extensively studied in computational models~\cite{kempter_hebbian_1999,
roberts_computational_1999, kistler_modeling_2000, van_rossum_stable_2000, kempter_intrinsic_2001,
rubin_equilibrium_2001, gerstner_mathematical_2002, burkitt_spike-timing-dependent_2004, rumsey_equalization_2004,
standage_computational_2007, gilson_stdp_2011}.

Other types of polarity have been observed experimentally, they are often neglected in theoretical studies of
STDP. For example, {\em Anti-Hebbian STDP} follows the opposite principles, whereby
pre-post pairings lead to depression, and  post-pre pairings to potentiation.
Such plasticities were observed experimentally in the striatum, see~\cite{fino_bidirectional_2005,
roberts_anti-hebbian_2010}.
Different types of STDP rules were analyzed in~\cite{PhysRevE.62.4077, cateau_stochastic_2003,
rumsey_equalization_2004, zou_kinetic_2007, roberts_anti-hebbian_2010, burbank_depression-biased_2012}.

The context, in general, is a single neuron receiving a large number of excitatory inputs subject to
STDP, leading to a Fokker-Plank approach~\cite{van_rossum_stable_2000, rubin_equilibrium_2001,
burkitt_spike-timing-dependent_2004}.
In particular, the importance of a single pairing is diluted over the large number of inputs in the mean-field limit,
whereas by definition STDP relies on the repetition of such correlated pairings.

The use of the pre-/postsynaptic spike correlation function~\cite{kempter_hebbian_1999,kempter_intrinsic_2001,
gerstner_mathematical_2002} was used to study the influence of STDP with high correlated inputs.
However, this method relies on the assumption that all pairs of pre- and postsynaptic spikes impact the synaptic
weight update.
Several studies have questioned this hypothesis~\cite{izhikevich_relating_2003,
morrison_spike-timing-dependent_2007, morrison_phenomenological_2008}, and its influence on the synaptic weight
dynamics has not been discussed in theoretical works, except in~\cite{burkitt_spike-timing-dependent_2004}.

Finally, most studies focus on excitatory inputs, whereas inhibitory synapses also exhibit STDP~\cite{
haas_spike-timing-dependent_2006,feldman_spike-timing_2012}, but few theoretical works exist~\cite{luz_effect_2014}.

Here we develop a theoretical study of a large class of rules, for a system with two neurons and a
single synapse.
This simple setting is used to test the influence of STDP on an excitatory and an inhibitory synapse, for three
different classes of pairing interactions leading to an extensive categorization of the different dynamics.
We show in particular that several interesting properties of the synaptic weight dynamics are lost when using classical
models with numerous excitatory inputs, leading to an underestimation of the role of STDP in learning systems.
\section{Theoretical analysis}
\label{sec:theoretical}

\begin{figure}[ht]\centerfloat\includegraphics{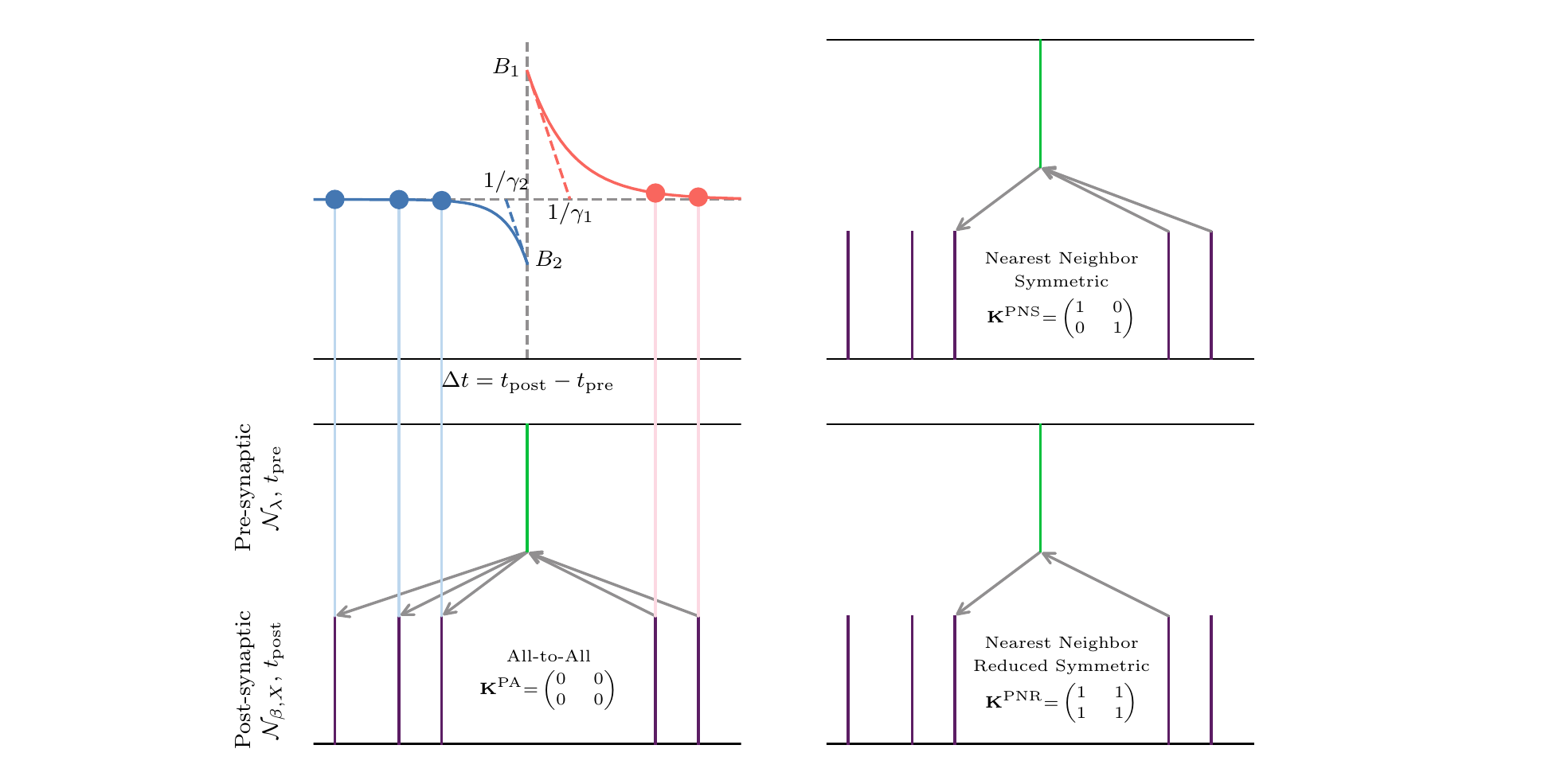}\caption{Synaptic Plasticity Kernels for Pair-Based Rules}
\label{figsub:pair}\end{figure}

\subsection{Spiking neurons and Poisson processes}
\label{secsec:defnot}

The spike train of the pre-synaptic neuron is represented by an homogeneous Poisson process ${\cal N}_{\lambda}$ with $\lambda{>}0$,
where $\delta_x$ is the Dirac measure at $x{\in}\R$, then
\[
{\cal N}_\lambda{=}\sum_{n{\ge}1} \delta_{t_{{\rm pre},n}}, \text{ with  }
0{\le}t_{{\rm pre},1}\le\cdots \le t_{{\rm pre},n}\le\cdots,
\]
with $t_{pre/post}$ the pre- and postsynaptic spike times.

In particular $\P\left(\rule{0mm}{5mm}{\cal N}_{\lambda}(t,t{+}\diff t){\neq} 0 \right)
       {=}\lambda\diff t{+}o(\diff t)$.

We define a stochastic process $(X(t))$ following leaky-integrate dynamics:
\begin{enumerate}
\item It decays exponentially to $0$ with a fixed exponential decay, set to $1$.
\item It is incremented by the synaptic weight $W{>}0$ at each pre-synaptic spike, i.e. at each instant of~${\cal N}_\lambda$.
\end{enumerate}
The firing mechanism of the postsynaptic neuron is driven by an {\em activation function function} $\beta$, when $X$ is $x$,
the output neuron fires at rate $\beta(x)$.
The sequence of instants of post-synaptic spikes $(t_{\textup{post},n})$ is a  point process
${\cal N}_{\beta,X}$ on $\R_+$ such that 
\[
\P\left.\left(\rule{0mm}{5mm}{\cal N}_{\beta,X}(t,t{+}\diff t){\neq} 0 \right| X(t){=}x\right)
{=}\beta(x)\diff t{+}o(\diff t).
\]
The process represents the membrane potential of the postsynaptic neuron in the case of an excitatory synapse.
For simplicity, we chose to differentiate between excitatory and inhibitory synapses at the level of the activation
function, instead of allowing for negative $W$.

Indeed, for an excitatory synapse, the activation function $\beta(x){=}\nu{+}\beta x$ is used, $\nu$ is the rate of the external
input  to the post-synaptic neuron, it models external noise. For inhibitory synapses, we consider $\beta(x){=}\max(\nu{-}\beta x,0)$.

\subsection{Pair-based STDP rules}

We study an important implementation of STDP referred to as pair-based rules.
For a pair $(t_{\rm pre},t_{\rm post})$ of instants of pre- and post-synaptic spikes, the synaptic weight update $\Delta W$
depends only on $\Delta t{=}t_{\rm post}{-}t_{\rm pre}$, as illustrated in Figure~\ref{figsub:pair}.
Most of STDP experimental studies are based on such pairing protocols, where pre- and post-synaptic spikes are repeated
with a fixed delay for a given number of evenly spaced pairings, see~\cite{morrison_phenomenological_2008,
bi_synaptic_1998, feldman_spike-timing_2012}.

An important choice for the model is to decide which pairings to take into account in the plasticity update.
A large choice of different schemes have been analyzed in the literature~\cite{morrison_phenomenological_2008}.
We have chosen to focus  on three versions, that are summarized in Figure~\ref{figsub:pair}:
\begin{itemize}
\item PA: \emph{All-to-all} pair-based model: all pairs of spikes are taken into account in the synaptic
plasticity rule.
\item PNS: {\em Nearest neighbor symmetric} model: whenever one neuron spikes, the synaptic weight is updated by only
taking into account the last spike of the other neuron.
\item PNR: {\em Nearest neighbor symmetric reduced} model: only consecutive pairs of spikes are used to update the
synaptic weight.
\end{itemize}

The synaptic weight update is therefore composed of the sum over relevant spikes, of an kernel $\Phi(\Delta t)$
known as the plasticity curve, here we chose an exponential kernel, given by,
\[
  \Phi(\Delta t) = \begin{cases}
          B_2\exp(\gamma_2 \Delta t) & \Delta t {<} 0, \\
          B_1\exp(-\gamma_1 \Delta t) & \Delta t {>} 0. \\
       \end{cases}
\]
where $B_1,B_2{\in}\R$ represents the amplitude of the STDP and $\gamma_1,\gamma_2>0$ the characteristic time of
interaction, see Figure~\ref{figsub:pair} (top).

All these pair-based rules can be represented by a system of the form
\begin{equation}
\label{eq:pbsystem}
\begin{cases}\diff X(t)  &= {-} X(t)\diff t+W(t{-})\mathcal{N}_{\lambda}(\diff t),\\
    \diff Z_{1}(t) &=   {-}\gamma_1 Z_{1}(t) \diff t\\
        &\hspace{10mm}+(B_1-K_{1,1}Z_{1}(t{-}))\mathcal{N}_{\lambda}(\diff t)\\
        &\hspace{10mm}-K_{1,2}Z_{1}(t{-})\mathcal{N}_{\beta,X}(\diff t),\\
    \diff Z_{2}(t) &=   {-}\gamma_2 Z_{2}(t) \diff t\\
    &\hspace{10mm}+ (B_2{-}K_{2,2}Z_{2}(t{-}))\mathcal{N}_{\beta,X}(\diff t)\\
    &\hspace{10mm}-K_{2,1}Z_{2}(t{-})\mathcal{N}_{\lambda}(\diff t),\\
    \diff W(t) &=  Z_{1}(t{-})\mathcal{N}_{\beta,X}(\diff t) +Z_{2}(t{-})\eps\mathcal{N}_{\lambda}(\diff t)
\end{cases}
\end{equation}
where  $\gamma_{1},\gamma_{2}{>}0$, $B_{1},B_{2}{\in}\R$,  ${\mathbf K}{=}(K_{ij},i,j{\in}\{1,2\}){\in}\{0,1\}^4$.

For the three pair-based STDP rules detailed, we have,
\[
    {\mathbf K}^{\rm PA}{=}\begin{pmatrix}
            0 & 0\\
            0 & 0
            \end{pmatrix},
    \quad     {\mathbf K}^{\rm PNS}{=} \begin{pmatrix}
            1 & 0\\
            0 & 1
            \end{pmatrix},
    \quad    {\mathbf K}^{\rm PNR}{=}\begin{pmatrix}
        1 & 1\\
        1 & 1
        \end{pmatrix}.
\]

Appendix~\ref{app:pi} provides more details on the model and equations.

\subsection{General formulation in a slow-fast system}

We consider that the processes $(X(t))$ and $(Z_1(t),Z_2(t))$ evolve on a fast time scale $t{\mapsto}t/\eps$ for some small $\eps{>}0$.
The increments of the variable $W$ are scaled with the parameter $\eps$, so that the variation on a bounded time-interval is $O(1)$, $(W_\eps(t))$ is described as the {\em slow} process.

An intuitive picture of approximation results used in this paper can be described as follows.
For $\eps$ small, on a short time interval, the slow process $(W_\eps(t))$ is almost constant, and, due to its
faster dynamics, the process $(X_\eps(t),Z_{1,\eps}(t),Z_{2,\eps}(t))$  is ``almost'' at its equilibrium distribution
associated to the current value of $W_\eps(t){\approx}w$. This is also the equilibrium of  the process $(X^w(t),Z_1^w(t),Z_2^w(t))$ such that
\begin{equation}
\label{def:pairstati}
\begin{cases}
    \diff X^w(t) \displaystyle &= {-}X^w(t)\diff t+w\mathcal{N}_{\lambda}(\diff t),\\
    \diff Z_1^w(t) \displaystyle &=   {-}\gamma Z_1^w(t) \diff t \\
    &\hspace{1cm} + (B_1-K_{1,1}Z_1^w(t{-}))\mathcal{N}_{\lambda}(\diff t)\\
    &\hspace{1cm}-K_{1,2}Z_1^w(t{-})\mathcal{N}_{\beta,X^w}(\diff t),\\
    \diff Z_2^w(t) \displaystyle &=   {-}\gamma Z_2^w(t) \diff t\\
    &\hspace{1cm} + (B_2-K_{2,2}Z_2^w(t{-}))\mathcal{N}_{\beta,X^w}(\diff t)\\
    &\hspace{1cm}-K_{2,1}Z_2^w(t{-})\mathcal{N}_{\lambda}(\diff t).
\end{cases}
\end{equation}
Classical results on Markov systems imply that there is unique stationary distribution $\Pi_w^{\mathbf K}$ on`
$\R_+{\times}\R^2$, for simplicity we will denote $\Pi_w^{\text{\rm PX}}{=}\Pi_w^{K^{\text{\rm PX}}}$,
see~\cite{robert_stochastic_2020_1}

Using averaging principle arguments,  the asymptotic dynamic of $(W_\eps(t))$ is given by the ODE,
\begin{align}\label{AsymW}
\frac{\diff w}{\diff t}(t) &= \int_{\R_+{\times}\R^2}(\lambda z_2{+}\beta(x)z_1)\Pi_{w(t)}^{\mathbf K}(\diff x,\diff z)\\
&=\E_{\Pi^{\mathbf K}_{w(t)}}\left[\lambda Z_{2} {+} \beta(X) Z_{1} \right].
\end{align}
A more rigorous development on this result is given in Appendix~\ref{app:slowfast}.
\begin{table*}
\centerfloat
\begin{tabular}{ ccc }
\hline
 \shortstack{\vspace{2em} LTD (Long Term Depression)} & \shortstack{\vspace{2em} $\forall w_0,$ $\displaystyle\lim_{t{\rightarrow}{+}\infty}w(t) = 0$}&
 \shortstack{\vspace{1em} \\ $p_{+\infty} < p_{\text{\rm bif}}$ \\ $p_0 \geq p_{\text{\rm bif}}$ \\ $p_{\text{\rm stable}} < p_{\text{\rm bif}}$ \vspace{1em}} \\ \hline
 \shortstack{\vspace{2em} LTP (Long Term Potentiation)} & \shortstack{\vspace{2em} $\forall w_0,$ $\displaystyle\lim_{t{\rightarrow}{+}\infty}w(t) = +\infty$}&
 \shortstack{\vspace{1em} \\ $p_{+\infty} \geq p_{\text{\rm bif}}$ \\ $p_0 < p_{\text{\rm bif}}$ \\ $p_{\text{\rm stable}} < p_{\text{\rm bif}}$ \vspace{1em}} \\ \hline
 \shortstack{\vspace{2em} UNSTABLE Fixed Point} & \shortstack{\vspace{1em} \\ $\exists w_{\textup{eq}}, \forall w_0 {<} w_{\textup{eq}},$  $\displaystyle\lim_{t{\rightarrow}{+}\infty}w(t) = 0$ \\ and $\forall w_0{>}w_{\textup{eq}},$ $\displaystyle\lim_{t{\rightarrow}{+}\infty}w(t) = +\infty$\vspace{1em}} &
 \shortstack{\vspace{1em} \\ $p_{+\infty} \geq p_{\text{\rm bif}}$ \\ $p_0 \geq p_{\text{\rm bif}}$ \\ $p_{\text{\rm stable}} < p_{\text{\rm bif}}$ \vspace{1em}} \\ \hline
 \shortstack{\vspace{2em} STABLE Fixed Point} & \shortstack{\vspace{1em} \\ $\exists w_{\textup{eq}},$ $\forall w_0, \displaystyle\lim_{t{\rightarrow}{+}\infty}w(t) =  w_{\textup{eq}}$\vspace{2em}}&
 \shortstack{\vspace{1em} \\ $p_{+\infty} < p_{\text{\rm bif}}$ \\ $p_0 < p_{\text{\rm bif}}$ \\ $p_{\text{\rm stable}} \geq p_{\text{\rm bif}}$ \vspace{1em}} \\ \hline
  \shortstack{\vspace{2em} MULTIPLE Fixed Point} & \shortstack{\vspace{1em} \\ Other behaviors\vspace{2em}}&
 \shortstack{Complementary set \vspace{2em}} \\  \hline
\end{tabular}
\caption{Different possible behaviors, theoretical definitions and numerical estimations}
\label{table:bif}
\end{table*}
\subsection{Computer simulations}
\label{secsec:methods}

To compare different dynamics, synapses and pairing schemes, we perform, for each set of parameters,
independent simulations and from this array of dynamics we compute several variables:
\begin{itemize}
\item The probability of diverging to infinity, $p_{+\infty}{=}\P\left(W_{\eps}(t){=}+\infty\right)$, approximated by
the proportion of simulations where the synaptic weight goes above $w_{\text{\rm max}}$.
\item The probability of converging  to $0$, $p_{0}{=}\P\left( W_{\eps}(t){=}0\right)$, approximated by
the proportion of simulations whose synaptic weight goes below $0$.
\item The probability to have a stable fixed point defined by the complementary probability
$p_{\text{\rm st}}{=}1{-}p_{+\infty}{-}p_{0}$.
\end{itemize}

Five different asymptotic behaviors for $w$, solution of~\eqref{AsymW} are defined using analytical asymptotic properties
in Table~\ref{table:bif}.
We define numerical approximates of these possible behaviors, depending on the values of $p_{+\infty}$, $p_{0}$ and
$p_{\text{\rm st}}$, and a fixed parameters $p_{\text{\rm bif}}$.
\section{Results}
\label{sec:results}

\begin{figure}[h!]
\centerfloat
\includegraphics{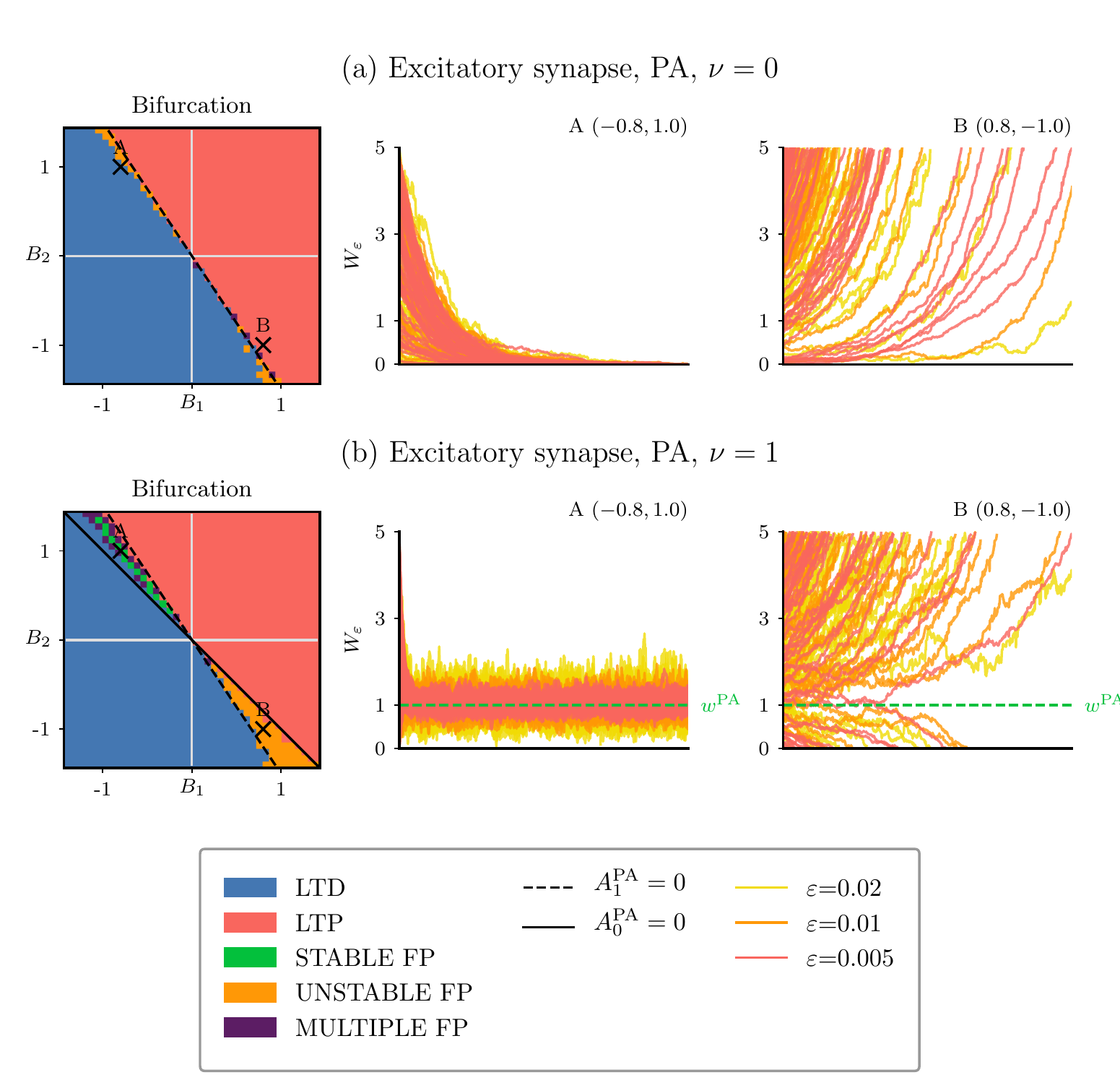}\caption{All-to-all pair-based STDP for an excitatory synapse}\label{fig:FigurePA}\end{figure}

In this framework, we study the asymptotic behavior of the dynamical system~\eqref{AsymW} for the three pair-based rules,
\[
 \frac{\diff w}{\diff t}(t){=} f^{\mathbf K}(w(t)), \text{ with }  f^{\mathbf K}(w){\steq{def}}\E_{\Pi^{{\mathbf K}}_{w}}\left[\lambda Z_{2} {+} \beta(X) Z_{1} \right],
 \] 
${\mathbf K}{\in}\{{\rm PA},{\rm PNS},{\rm PNR}\}$.

We will show that the synaptic weight $w(t)$ usually end up having one of three different asymptotic behaviors, which all
have a biological interpretation:
\begin{itemize}
\item Convergence of $w(t)$ towards $0$, which corresponds the disconnection (or pruning) of the synapse: the presynaptic neuron loses
its ability to influence the postsynaptic neuron.
\item Divergence of $w(t)$ to infinity, leading an unstable system which, in a biological system, will be stopped by saturation mechanisms.
\item Convergence to a non null value $w_{eq}$, resulting in self-sustained activity, i.e., pre- and postsynaptic activities coupled with
STDP are sufficient to have a bounded stable synaptic weight.
\end{itemize}

\subsection{Stability and divergence depends on the polarity of STDP}
\label{secsec:exc}

Starting with the all-to-all scheme for an excitatory synapse, i.e.  $\beta(x){=}\nu{+}\beta x$,   we have,
\[
    \frac{\diff w}{\diff t}(t) = f^{\text{\rm PA}}(w) =A^{\text{\rm PA}}_0{+}A^{\text{\rm PA}}_{1}w = A^{\text{\rm PA}}_1\left(w{-}w^{\text{\rm PA}}\right).
\]
with
\begin{multline*}
    A^{\text{\rm PA}}_{0} \steq{def} \nu\lambda\left(\frac{B_1}{\gamma_1}{+}\frac{B_2}{\gamma_2}\right),\\
    A^{\text{\rm PA}}_{1} \steq{def} \beta\lambda^2\left(\frac{B_1}{\gamma_1}{+}\frac{B_2}{\gamma_2}{+}\frac{B_1}{\lambda(1{+}\gamma_1)}\right) \text{ and } \\
    w^{\text{\rm PA}}\steq{def} {-}A^{\text{\rm PA}}_{0} / A^{\text{\rm PA}}_{1}.
\end{multline*}
The calculation is detailed in Appendix~\ref{app:proofPAEXC}. The signs of  $A^{\text{\rm PA}}_0$ and $A^{\text{\rm PA}}_1$  determine in fact the asymptotic
behavior of $w$. We study the impact of $B_1$ and $B_2$ with, or without, external input rate $\nu$ on the dynamics in~Figure~\ref{fig:FigurePA}.

If $\nu{=}0$, then $w^{\text{\rm PA}}{=}0$. Without external input $\nu$, the synaptic weights cannot converge to a positive stable solution.

\begin{itemize}
    \item If $A^{\text{\rm PA}}_1{<}0$, $(w(t))$ converges to $0$, as shown by the blue region of
Figure~\ref{fig:FigurePA} (a),  with some examples of dynamics at points $(A)$ and $(B)$ below.
    \item If $A^{\text{\rm PA}}_1{>}0$, $(w(t))$ diverges to $+\infty$, the red region of Figure~\ref{fig:FigurePA} (top) and example $(C)$.
\end{itemize}

If $\nu{>}0$, $w^{\text{\rm PA}}$ is a positive fixed point.  This gives two new behaviors in the bifurcation map, see Figure~\ref{fig:FigurePA} (b).

\begin{itemize}
    \item If $A^{\text{\rm PA}}_1{>}0$ and $A^{\text{\rm PA}}_0{<}0$, the fixed point is unstable (orange region), the example (B) shows
that in that case, the dynamics depends on the initial value of synaptic weight.  It diverges to $+\infty$ if starting above $w^{\text{\rm PA}}$, and converges to $0$ otherwise. 
    \item If $A^{\text{\rm PA}}_1{<}0$ and $A^{\text{\rm PA}}_0{>}0$, the fixed point is stable (green region) and all  simulations
converge to $w^{\text{\rm PA}}$ independently of the initial point. See example (C).
\end{itemize}
\subsection{Influence of pairing scheme}
\label{secsec:excPN}

\begin{figure}\centerfloat\includegraphics{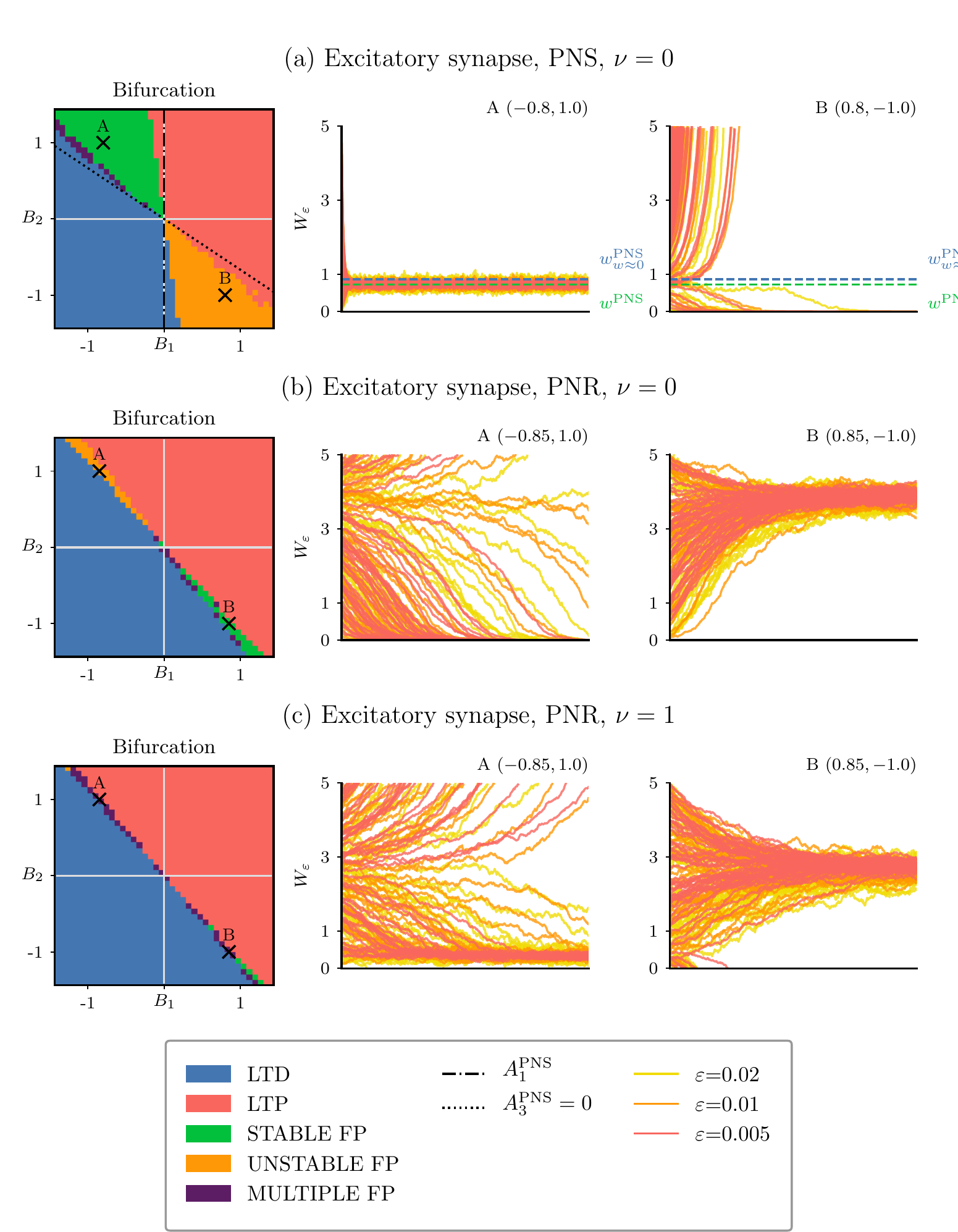}\caption{Different pairing schemes leads to diverse dynamics}\label{fig:FigurePIBIF}\end{figure}

\subsubsection*{Nearest neighbor symmetric STDP}

For nearest neighbor symmetric STDP with $\beta(x){=}\nu{+}\beta x$, we derive in Appendix~\ref{app:proofPNSEXC_f}, the associated dynamical system,
\[
    \frac{\diff w}{\diff t}(t) = f^{\text{\rm PNS}}(w) \steq{def} A^{\text{\rm PNS}}_0 {+} A^{\text{\rm PNS}}_1w {+} A^{\text{\rm PNS}}_2 h^{\text{\rm PNS}}(w),
\]
with,
\begin{align*}
    A^{\text{\rm PNS}}_0 &\steq{def} \frac{\nu\lambda}{\lambda{+}\gamma_1} B_1 {+} \frac{\nu\lambda}{\nu{+}\gamma_2} B_2, \\
    A^{\text{\rm PNS}}_1 &\steq{def} \lambda \beta \frac{1{+}\lambda}{1{+}\lambda{+}\gamma_1} B_1,\ A^{\text{\rm PNS}}_2{=}\lambda B_2,
\end{align*}
and
\begin{multline*}
    h^{\text{\rm PNS}}(w)\steq{def}\gamma_2 \int_{\R_+}e^{-\gamma_2\tau}\left(1- \exp\left(\rule{0mm}{4mm}{-}\nu \tau\right.\right.\\
    \left.\left.{-}\lambda \hspace{-1mm}\int_{0}^{\tau} \hspace{-2mm} \left(1{-}\exp\left(-\beta w \left(1{-}e^{s-\tau}\right)\right)\right)\diff s\right.\right.\\
    \left.\left.\rule{0mm}{4mm}{-}\lambda \int_{-\infty}^{0} \hspace{-2mm} \left(1{-}\exp\left(-\beta w \left(1{-}e^{-\tau}\right)e^s\right)\right)\diff s\right)\right) \diff \tau - \frac{\nu}{\nu{+}\gamma_2}.
\end{multline*}

The asymptotic behavior of $(w(t))$ can be analyzed rigorously in this case, details in Appendix~\ref{app:proofPNSEXC_dynamics}.

If $\nu{=}0$, let
\[
    A^{\text{\rm PNS}}_3 \steq{def} {f^{\text{\rm PNS}}}^{\prime}(0){=} \lambda \beta\left(\frac{1+\lambda}{1+\lambda+\gamma_1} B_1{+} \frac{\lambda}{\gamma_2} B_2\right)
\]

\begin{itemize}
\item If $A^{\text{\rm PNS}}_1{<}0$ and $A^{\text{\rm PNS}}_3{<}0$,  $(w(t))$ converges to $0$ in finite time (blue).
\item If $A^{\text{\rm PNS}}_1{>}0$ and $A^{\text{\rm PNS}}_3{>}0$, The system diverges to infinity when both parameters are positive (red).
\item If $A^{\text{\rm PNS}}_1{<}0$ and $A^{\text{\rm PNS}}_3{>}0$, a stable fixed point $w^{\text{\rm PNS}}$ exists  (green),  see example $A $  $(B_1,B_2){=}(-0.8,1)$.
\item If $A^{\text{\rm PNS}}_1{>}0$ and $A^{\text{\rm PNS}}_3{<}0$,  an unstable fixed point $w^{\text{\rm PNS}}$ exists (orange), example B.
\end{itemize}

In Appendix~\ref{app:proofPNSEXC_dynamics}, we prove the existence of the fixed point $w^{\text{\rm PNS}}$ and provided a numerical estimation in Figure~\ref{fig:FigurePIBIF}(a).
We compute an approximation of $w^{\text{\rm PNS}}$ when $w{\approx}0$ in~\ref{app:proofPNSEXC_dynamics} is given,
Figure~\ref{fig:FigurePIBIF}(a) shows a comparison with numerical experiments.

The picture is similar for the case $\nu{>}0$, with slightly different conditions (Appendix~\ref{app:proofPNSEXC_dynamics} and Figure~\ref{fig:FigurePIBIF_PNS}).

\medskip
\noindent
{\bf Discussion}. Nearest neighbor symmetric STDP has significant differences with the all-to-all scheme.
First, a positive stable (or unstable) fixed point may exist in the absence of external noise. 
The condition on $A^{\text{\rm PNS}}_1$ is a condition on $B_1$ only.  If $B_1{<}0$ the system either converges to $0$ or to a positive fixed point, and similarly when $B_1{>}0$. The all-to-all case does not exhibit such a simple behavior, because $A^{\text{\rm PA}}_0$ and $A^{\text{\rm PA}}_1$
both depend on $B_1$ and $B_2$.

\subsubsection*{Nearest neighbor symmetric reduced STDP}

A theoretical study of $(w(t))$ solution of~\eqref{AsymW}  with $\beta(x){=}\nu{+}\beta x$ is possible, but more involved than
for PA and PNS. Some indications are given in the Appendix~\ref{app:proofPNREXC}.
Computer simulations were done using this scheme and the results are illustrated in Figure~\ref{fig:FigurePIBIF}(b) and (c).
Surprisingly, we observe two different dynamics depending on the values of $\nu$.

For $\nu{=}0$, there exists a (narrow) range of parameters in the Hebbian region (bottom right) where a stable fixed
point occurs, see example (B) in Figure~\ref{fig:FigurePIBIF}(b).
Symmetrically, an unstable fixed point seems to exist in the anti-Hebbian region
(top left) and example (A).

For $\nu{>}0$, a second fixed point appears leading to more complex behaviors characterized by the presence of a
stable and an unstable fixed point at the same time~Figure~\ref{fig:FigurePIBIF}(c).
If the stable fixed point is lower than the unstable one, see example (A) and (top left) in
Figure~\ref{fig:FigurePIBIF}(c), the synaptic weight either converges to a non null value or diverges to infinity.
For Hebbian parameters (bottom right), the situation is reversed, see example (B) in Figure~\ref{fig:FigurePIBIF}(c).
The spectrum of values with this complex behavior narrows when $\nu$ is increasing.
In particular, for large values of $\nu$, only a perfect balance in the parameters may lead to other behaviors than whole depression or potentiation.
We have studied this influence of $\nu$, on the dynamics, for $B_1$ and $B_2$ constant in Figure~\ref{fig:EXTERNAL}.

\medskip
\noindent {\bf Discussion}. 
There are several differences of interest with the two other STDP pair-based rules for an excitatory synapse.
First, for all-to-all and nearest neighbor symmetric pairings at an excitatory synapse, the stable fixed point
only appears for anti-Hebbian parameters, $B_1{<}0$, whereas an unstable one exists for
Hebbian STDP, $B_1{>}0$. With nearest neighbor symmetric reduced STDP, we have numerically shown that a more
complex behavior with several fixed points may occur. 

Second, the nearest neighbor symmetric reduced STDP needs an almost exact balance of the parameters to
enable convergence of the system toward a fixed point.

Table~\ref{table:compPIexc} gathers up all results for an excitatory synapse.
\subsection{All-to-all  STDP with an  inhibitory synapse}
\label{secsec:inh}

\begin{figure}[h!]\centerfloat\includegraphics{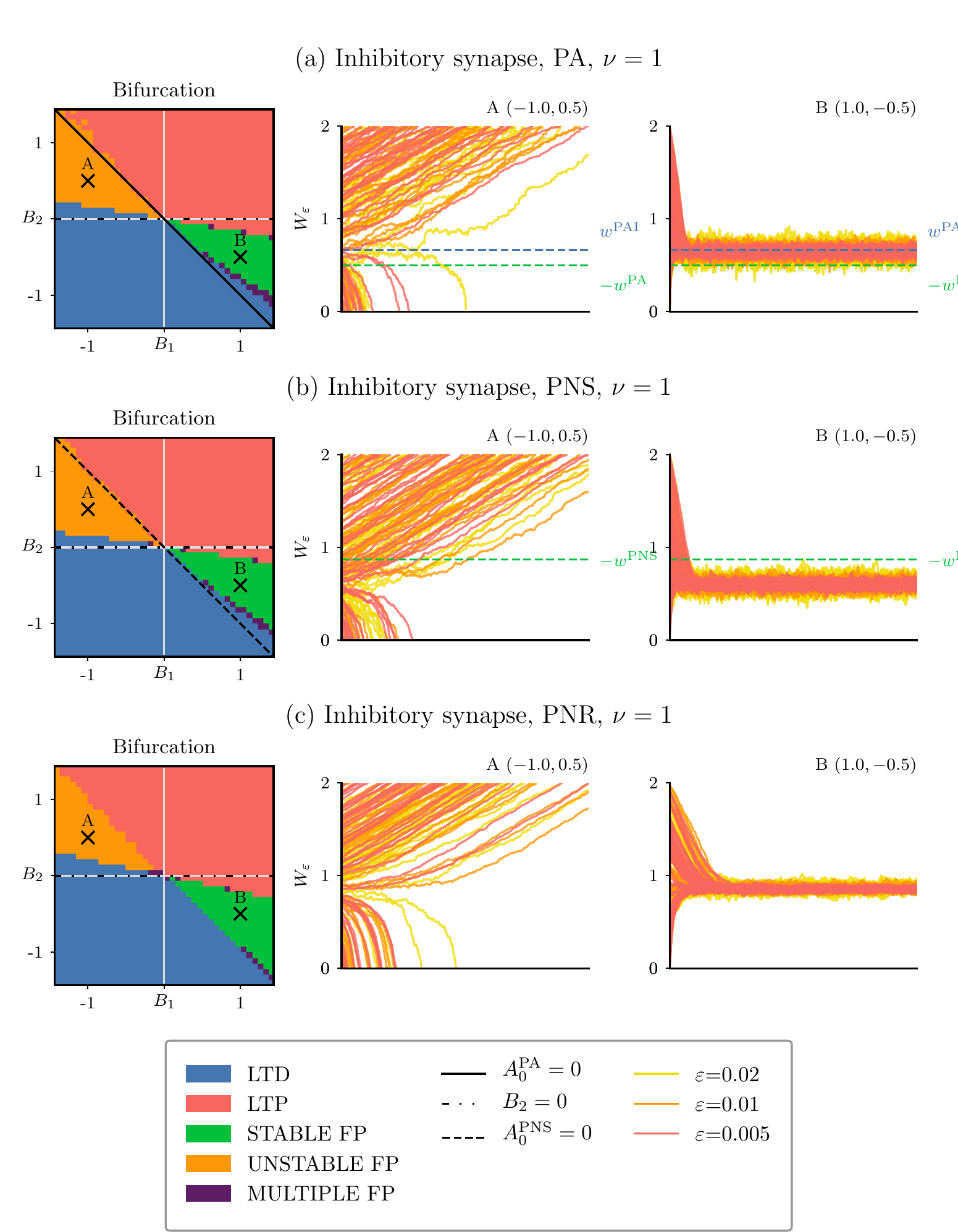}\caption{Pair-based STDP for an inhibitory synapse}\label{fig:FigurePII}\end{figure}

We now study the dynamics~\eqref{AsymW} of the synaptic weight for an inhibitory synapse, i.e. when $\beta(x){=}(\nu{-}\beta x)^+$.

Computations of $f^{\text{\rm PA}}$ are detailed in Appendix~\ref{app:proofPAINH}.
We restrict our study to two cases.

For small $w$,
\begin{align*}
    \frac{\diff w}{\diff t}(t) &= f^{\text{\rm PA}}(w) = A^{\text{\rm PA}}_0 {-} A^{\text{\rm PA}}_{1}w{+} o(w) \\
        &= {-}A^{\text{\rm PA}}_1\left(w{+}w^{\text{\rm PA}}\right) {+} o(w).
\end{align*}
with $A^{\text{\rm PA}}_{0/1}$ defined before.

When $w \geq \nu /\beta$, we have
\[
\frac{\diff w}{\diff t}(t) = \frac{A^{\text{\rm PAI}}}{w(t)^{\lambda+\gamma_1}}\left(1 {+} \eta^{\text{\rm PAI}}\left[\frac{w(t)}{w^{\text{\rm PAI}}} \right]^{\gamma_1} \right)
\]
with
\[
    A^{\text{\rm PAI}} \steq{def} \left[ \frac{\nu}{\beta} \right]^{\lambda+\gamma_1}\frac{c(\lambda) B_1 \nu}{\left( \lambda+\gamma_1\right)\left( \lambda{+}\gamma_1{+}1\right)},
\]
\[
    w^{\text{\rm PAI}} \steq{def} \frac{\beta}{\nu} \left(\left| \frac{B_2}{B_1} \right|\frac{\left( \lambda{+}\gamma_1\right)\left( \lambda+\gamma_1{+}1\right)}{\gamma_2\left( \lambda+1\right)}\right)^{1/\gamma_1},
\]
\[
    \eta^{\text{\rm PAI}} \steq{def} \left| \frac{B_2}{B_1} \right| \frac{B_1}{B_2}.
\]

For stability properties, the two relevant parameters are $A^{\text{\rm PA}}_{0}$ and $B_2$.

\begin{itemize}
\item For $B_2{>}0$ and $A^{\text{\rm PA}}_{0}{>}0$, the synaptic weight diverges to infinity (in red),
\item For $B_2{<}0$ and $A^{\text{\rm PA}}_{0}{<}0$, it converges to $0$ in finite time (in blue).
\item For $B_2{>}0$ and $A^{\text{\rm PA}}_{0}{<}0$, there is an
unstable fixed point (orange, example A).
\item For $B_2{<}0$ and $A^{\text{\rm PA}}_{0}{>}0$, the system exhibits a stable equilibrium (green, example B).
\end{itemize}

We note here an inversion with the properties observed for the excitatory synapse, where only anti-Hebbian
STDP led to a stable fixed point, compared to the inhibitory case where only Hebbian STDP elicits this type
of behavior.

Moreover, $A^{\text{\rm PA}}_{0}{=}0$ corresponds to the line $w^{\text{\rm PA}}{=}0$, suggesting that an
important parameter for the classification of behavior is the range of parameters where $w^{\text{\rm PA}}{=}0$ for
the excitatory case.

This analysis is completed with the other schemes in Figure~\ref{fig:FigurePII}(b) for PNS and Figure~\ref{fig:FigurePII}(c)
for PNR. The dynamics are similar to the all-to-all case for this range of parameters, only the values of the
fixed points seems to change (compare B for the three cases).
For the nearest neighbor symmetric STDP, we also plotted the line $w^{\text{\rm PNS}}{=}0$, as it could be
related to the change of dynamics following the analysis of the all-to-all case.

All these behaviors are gathered in Table~\ref{table:compPIinh}.
It is striking that the pairing scheme does not seem to have a decisive impact on the dynamics for
an inhibitory synapse, contrarily to the case of an excitatory synapse.
This may be due to the fact that for the inhibitory case, we need to have a constant external input in order
to have spikes. We note here that, for the sake of simplicity, we only tested cases where $\nu{=}1$.
\section{Conclusion}
\label{sec:conclu}

We have developed a simple and rigorous analysis of synaptic weight dynamics via
a slow-fast approximation  and numerical simulations.
For an excitatory synapse, anti-Hebbian STDP can lead to a stable fixed point, with some slight variations
depending on the pairing scheme used.
In particular, for all-to-all STDP rules, a fixed point exists only for positive external rate $\nu$, whereas
for nearest symmetric reduced scheme, at least two fixed points exists for balanced STDP rules.
Moreover, for an inhibitory synapse, numerical arguments showed that
all schemes were similar, with the existence of a stable fixed point for Hebbian STDP.

In a learning paradigm, a subset of correlated neurons are able to repeatedly trigger an action potential of the postsynaptic neuron, even in the presence of noise. It is not a surprise then if this regime of activity led to the
most interesting behaviors for the synaptic weight dynamics of our study.
Indeed, when the influence of a single neuron (or by extension a group of correlated neurons) is not negligible compared to the rest of inputs (when $\lambda w{>}\nu$), the asymptotic behavior of the synaptic weight highly depends on the polarity of the STDP curve and the pairing scheme.

On the contrary, when the impact of the presynaptic neuron spikes is lost in the external noise (consistent with a
large number of external uncorrelated inputs, $\lambda w{<}\nu$), pairing schemes do not influence the type of
dynamics observed.
Indeed, the influence of `direct' and `repetitive' pairings is lost in the large noise limit: in mean-field models, the synaptic weight dynamics is essentially driven by the {\em mean} synaptic weight,
see~\cite{babadi_stability_2016}.

This work highlights the fact that the choice of spikes to take into account in STDP is an essential part of the modeling process.
In particular, this conclusion should apply to more complex pairing schemes such as triplets rules~\cite{pfister_beyond_2006, babadi_stability_2016} or more complex calcium-based rules.

If this article focuses on a single synapse dynamics, its conclusions can be used to explain some of the  results from the literature on the influence of STDP in recurrent networks~\cite{burkitt_spike-timing-dependent_2007, gilson_stdp_2010, triplett_emergence_2018}. \cite{locherbach_large_2017} studies short-term plasticity in a large network
and~\cite{lucken_noise-enhanced_2016}  the noise-enhanced coupling of two excitatory neurons subject to STDP, which can be extended to the formation of multiclusters in adaptive networks~\cite{berner_multiclusters_2019}.
It would be challenging to extend our results to large stochastic networks with plastic synapses where theoretical studies are scarce.
Multi-dimensional auto-exciting/inhibiting processes are an important tool in this context. In particular  Hawkes processes, see~\cite{reynaud-bouret_inference_2013,costa_renewal_2018}. This is a promising approach toward a better understanding of learning in adaptive neural systems.

\printbibliography

@article{lucken_noise-enhanced_2016,
	title = {Noise-enhanced coupling between two oscillators with long-term plasticity},
	volume = {93},
	url = {https://link.aps.org/doi/10.1103/PhysRevE.93.032210},
	doi = {10.1103/PhysRevE.93.032210},
	number = {3},
	urldate = {2021-09-28},
	journal = {Physical Review E},
	author = {Lucken, Leonhard and Popovych, Oleksandr V. and Tass, Peter A. and Yanchuk, Serhiy},
	month = mar,
	year = {2016},
	note = {Publisher: American Physical Society},
	pages = {032210},
}

@article{roberts_anti-hebbian_2010,
	title = {Anti-hebbian spike-timing-dependent plasticity and adaptive sensory processing},
	volume = {4},
	issn = {1662-5188},
	doi = {10.3389/fncom.2010.00156},
	language = {english},
	journal = {Frontiers in Computational Neuroscience},
	author = {Roberts, Patrick D. and Leen, Todd K.},
	year = {2010},
	pmid = {21228915},
	pmcid = {PMC3018773},
	pages = {156},
}

@article{kempter_intrinsic_2001,
	title = {Intrinsic stabilization of output rates by spike-based {Hebbian} learning},
	volume = {13},
	issn = {0899-7667},
	doi = {10.1162/089976601317098501},
	language = {english},
	number = {12},
	journal = {Neural Computation},
	author = {Kempter, R. and Gerstner, W. and van Hemmen, J. L.},
	month = dec,
	year = {2001},
	pmid = {11705408},
	pages = {2709--2741},
}

@article{gerstner_mathematical_2002,
	title = {Mathematical formulations of {Hebbian} learning},
	volume = {87},
	issn = {0340-1200},
	doi = {10.1007/s00422-002-0353-y},
	language = {english},
	number = {5-6},
	journal = {Biological Cybernetics},
	author = {Gerstner, Wulfram and Kistler, Werner M.},
	month = dec,
	year = {2002},
	pmid = {12461630},
	pages = {404--415},
}

@article{cateau_stochastic_2003,
	title = {A stochastic method to predict the consequence of arbitrary forms of spike-timing-dependent plasticity},
	volume = {15},
	issn = {0899-7667},
	doi = {10.1162/089976603321192095},
	language = {english},
	number = {3},
	journal = {Neural Computation},
	author = {C{\^a}teau, Hideyuki and Fukai, Tomoki},
	month = mar,
	year = {2003},
	pmid = {12620159},
	pages = {597--620},
}

@article{burkitt_spike-timing-dependent_2004,
	title = {Spike-timing-dependent plasticity: the relationship to rate-based learning for models with weight dynamics determined by a stable fixed point},
	volume = {16},
	issn = {0899-7667},
	shorttitle = {Spike-timing-dependent plasticity},
	doi = {10.1162/089976604773135041},
	language = {english},
	number = {5},
	journal = {Neural Computation},
	author = {Burkitt, Anthony N. and Meffin, Hamish and Grayden, David B.},
	month = may,
	year = {2004},
	pmid = {15070504},
	pages = {885--940},
}

@article{rumsey_equalization_2004,
	title = {Equalization of synaptic efficacy by activity- and timing-dependent synaptic plasticity},
	volume = {91},
	issn = {0022-3077},
	doi = {10.1152/jn.00900.2003},
	language = {english},
	number = {5},
	journal = {Journal of Neurophysiology},
	author = {Rumsey, Clifton C. and Abbott, L. F.},
	month = may,
	year = {2004},
	pmid = {14681332},
	pages = {2273--2280},
}

@inproceedings{pfister_beyond_2006,
	title = {Beyond {Pair}-{Based} {STDP}: a {Phenomenological} {Rule} for {Spike} {Triplet} and {Frequency} {Effects}},
	volume = {18},
	shorttitle = {Beyond {Pair}-{Based} {STDP}},
	url = {https://proceedings.neurips.cc/paper/2005/hash/a4666cd9e1ab0e4abf05a0fb232f4ad3-Abstract.html},
	urldate = {2021-09-16},
	booktitle = {Advances in {Neural} {Information} {Processing} {Systems}},
	publisher = {MIT Press},
	author = {Pfister, Jean-pascal and Gerstner, Wulfram},
	year = {2006},
}

@article{standage_computational_2007,
	title = {Computational consequences of experimentally derived spike-time and weight dependent plasticity rules},
	volume = {96},
	issn = {0340-1200},
	doi = {10.1007/s00422-007-0152-6},
	language = {english},
	number = {6},
	journal = {Biological Cybernetics},
	author = {Standage, Dominic and Jalil, Sajiya and Trappenberg, Thomas},
	month = jun,
	year = {2007},
	pmid = {17468882},
	pages = {615--623},
}

@article{zou_kinetic_2007,
	title = {Kinetic models of spike-timing dependent plasticity and their functional consequences in detecting correlations},
	volume = {97},
	issn = {0340-1200},
	doi = {10.1007/s00422-007-0155-3},
	language = {english},
	number = {1},
	journal = {Biological Cybernetics},
	author = {Zou, Quan and Destexhe, Alain},
	month = jul,
	year = {2007},
	pmid = {17530277},
	pages = {81--97},
}

@article{gilson_stdp_2011,
	title = {{STDP} allows fast rate-modulated coding with {Poisson}-like spike trains},
	volume = {7},
	issn = {1553-7358},
	doi = {10.1371/journal.pcbi.1002231},
	language = {english},
	number = {10},
	journal = {PLoS computational biology},
	author = {Gilson, Matthieu and Masquelier, Timoth{\'e}e and Hugues, Etienne},
	month = oct,
	year = {2011},
	pmid = {22046113},
	pmcid = {PMC3203056},
	pages = {e1002231},
}

@article{burbank_depression-biased_2012,
	title = {Depression-biased reverse plasticity rule is required for stable learning at top-down connections},
	volume = {8},
	issn = {1553-7358},
	doi = {10.1371/journal.pcbi.1002393},
	language = {english},
	number = {3},
	journal = {PLoS computational biology},
	author = {Burbank, Kendra S. and Kreiman, Gabriel},
	year = {2012},
	pmid = {22396630},
	pmcid = {PMC3291526},
	pages = {e1002393},
}

@article{luz_effect_2014,
	title = {The {Effect} of {STDP} {Temporal} {Kernel} {Structure} on the {Learning} {Dynamics} of {Single} {Excitatory} and {Inhibitory} {Synapses}},
	volume = {9},
	issn = {1932-6203},
	url = {https://www.ncbi.nlm.nih.gov/pmc/articles/PMC4085044/},
	doi = {10.1371/journal.pone.0101109},
	number = {7},
	urldate = {2021-09-16},
	journal = {PLoS ONE},
	author = {Luz, Yotam and Shamir, Maoz},
	month = jul,
	year = {2014},
	pmid = {24999634},
	pmcid = {PMC4085044},
	pages = {e101109},
}

@article{gilson_stdp_2010,
	title = {{STDP} in {Recurrent} {Neuronal} {Networks}},
	volume = {4},
	issn = {1662-5188},
	url = {https://www.frontiersin.org/articles/10.3389/fncom.2010.00023/full},
	doi = {10.3389/fncom.2010.00023},
	language = {English},
	urldate = {2019-11-08},
	journal = {Frontiers in Computational Neuroscience},
	author = {Gilson, Matthieu and Burkitt, Anthony and Van Hemmen, Leo J.},
	year = {2010},
}

@article{burkitt_spike-timing-dependent_2007,
	title = {Spike-timing-dependent plasticity for neurons with recurrent connections},
	volume = {96},
	issn = {0340-1200},
	doi = {10.1007/s00422-007-0148-2},
	language = {english},
	number = {5},
	journal = {Biological Cybernetics},
	author = {Burkitt, A. N. and Gilson, M. and van Hemmen, J. L.},
	month = may,
	year = {2007},
	pmid = {17415586},
	pages = {533--546},
}

@article{triplett_emergence_2018,
	title = {Emergence of spontaneous assembly activity in developing neural networks without afferent input},
	volume = {14},
	issn = {1553-734X},
	url = {https://www.ncbi.nlm.nih.gov/pmc/articles/PMC6161857/},
	doi = {10.1371/journal.pcbi.1006421},
	number = {9},
	urldate = {2021-09-16},
	journal = {PLoS Computational Biology},
	author = {Triplett, Marcus A. and Avitan, Lilach and Goodhill, Geoffrey J.},
	month = sep,
	year = {2018},
	pmid = {30265665},
	pmcid = {PMC6161857},
	pages = {e1006421},
}

@article{babadi_stability_2016,
	title = {Stability and {Competition} in {Multi}-spike {Models} of {Spike}-{Timing} {Dependent} {Plasticity}},
	volume = {12},
	issn = {1553-7358},
	doi = {10.1371/journal.pcbi.1004750},
	language = {english},
	number = {3},
	journal = {PLoS computational biology},
	author = {Babadi, Baktash and Abbott, L. F.},
	month = mar,
	year = {2016},
	pmid = {26939080},
	pmcid = {PMC4777380},
	pages = {e1004750}
}

@article{bi_synaptic_1998,
	title = {Synaptic {Modifications} in {Cultured} {Hippocampal} {Neurons}: {Dependence} on {Spike} {Timing}, {Synaptic} {Strength}, and {Postsynaptic} {Cell} {Type}},
	volume = {18},
	copyright = {Copyright {\textcopyright} 1998 Society for Neuroscience},
	issn = {0270-6474, 1529-2401},
	shorttitle = {Synaptic {Modifications} in {Cultured} {Hippocampal} {Neurons}},
	url = {http://www.jneurosci.org/content/18/24/10464},
	doi = {10.1523/JNEUROSCI.18-24-10464.1998},
	language = {english},
	number = {24},
	urldate = {2019-04-23},
	journal = {Journal of Neuroscience},
	author = {Bi, Guo-qiang and Poo, Mu-ming},
	month = dec,
	year = {1998},
	pmid = {9852584},
	pages = {10464--10472}
}

@article{costa_renewal_2018,
	title = {Renewal in {Hawkes} processes with self-excitation and inhibition},
	url = {http://arxiv.org/abs/1801.04645},
	urldate = {2019-06-18},
	journal = {arXiv:1801.04645 [math]},
	author = {Costa, Manon and Graham, Carl and Marsalle, Laurence and Tran, Viet Chi},
	month = jan,
	year = {2018},
	note = {arXiv: 1801.04645}
}

@article{feldman_spike-timing_2012,
	title = {The spike-timing dependence of plasticity},
	volume = {75},
	issn = {1097-4199},
	doi = {10.1016/j.neuron.2012.08.001},
	language = {english},
	number = {4},
	journal = {Neuron},
	author = {Feldman, Daniel E.},
	month = aug,
	year = {2012},
	pmid = {22920249},
	pmcid = {PMC3431193},
	pages = {556--571}
}

@article{fino_bidirectional_2005,
	title = {Bidirectional activity-dependent plasticity at corticostriatal synapses},
	volume = {25},
	issn = {1529-2401},
	doi = {10.1523/JNEUROSCI.4476-05.2005},
	language = {english},
	number = {49},
	journal = {The Journal of Neuroscience: The Official Journal of the Society for Neuroscience},
	author = {Fino, Elodie and Glowinski, Jacques and Venance, Laurent},
	month = dec,
	year = {2005},
	pmid = {16339023},
	pages = {11279--11287}
}

@book{gerstner_spiking_2002,
	title = {Spiking {Neuron} {Models}: {Single} {Neurons}, {Populations}, {Plasticity}},
	isbn = {978-0-521-89079-3},
	shorttitle = {Spiking {Neuron} {Models}},
	language = {english},
	publisher = {Cambridge University Press},
	author = {Gerstner, Wulfram and Kistler, Werner M.},
	month = aug,
	year = {2002},
	note = {Google-Books-ID: Rs4oc7HfxIUC}
}

@article{gilbert_amplitude_1960,
	title = {Amplitude {Distribution} of {Shot} {Noise}},
	volume = {39},
	copyright = {{\textcopyright} 1960 The Bell System Technical Journal},
	issn = {1538-7305},
	url = {https://onlinelibrary.wiley.com/doi/abs/10.1002/j.1538-7305.1960.tb01603.x},
	doi = {10.1002/j.1538-7305.1960.tb01603.x},
	language = {english},
	number = {2},
	urldate = {2018-09-27},
	journal = {Bell System Technical Journal},
	author = {Gilbert, E. N. and Pollak, H. O.},
	month = mar,
	year = {1960},
	pages = {333--350}
}

@article{izhikevich_relating_2003,
	title = {Relating {STDP} to {BCM}},
	volume = {15},
	issn = {0899-7667},
	doi = {10.1162/089976603321891783},
	language = {english},
	number = {7},
	journal = {Neural Computation},
	author = {Izhikevich, Eugene M. and Desai, Niraj S.},
	month = jul,
	year = {2003},
	pmid = {12816564},
	pages = {1511--1523}
}

@article{kempter_hebbian_1999,
	title = {Hebbian learning and spiking neurons},
	volume = {59},
	url = {https://link.aps.org/doi/10.1103/PhysRevE.59.4498},
	doi = {10.1103/PhysRevE.59.4498},
	number = {4},
	urldate = {2019-11-08},
	journal = {Physical Review E},
	author = {Kempter, Richard and Gerstner, Wulfram and van Hemmen, J. Leo},
	month = apr,
	year = {1999},
	pages = {4498--4514}
}

@article{kistler_modeling_2000,
	title = {Modeling {Synaptic} {Plasticity} in {Conjunction} with the {Timing} of {Pre}- and {Postsynaptic} {Action} {Potentials}},
	volume = {12},
	issn = {0899-7667},
	doi = {10.1162/089976600300015844},
	number = {2},
	journal = {Neural Computation},
	author = {Kistler, Werner M. and Hemmen, J. Leo van},
	month = feb,
	year = {2000},
	pages = {385--405}
}

@article{berner_multiclusters_2019,
	title = {Multiclusters in {Networks} of {Adaptively} {Coupled} {Phase} {Oscillators}},
	volume = {18},
	url = {https://epubs.siam.org/doi/abs/10.1137/18M1210150},
	doi = {10.1137/18M1210150},
	number = {4},
	urldate = {2021-09-28},
	journal = {SIAM Journal on Applied Dynamical Systems},
	author = {Berner, Rico and Schöll, Eckehard and Yanchuk, Serhiy},
	month = jan,
	year = {2019},
	note = {Publisher: Society for Industrial and Applied Mathematics},
	pages = {2227--2266},
}

@article{locherbach_large_2017,
	title = {Large deviations for cascades of diffusions arising in oscillating systems of interacting {Hawkes} processes},
	url = {http://arxiv.org/abs/1709.09356},
	urldate = {2019-06-19},
	journal = {arXiv:1709.09356 [math]},
	author = {L{\"o}cherbach, Eva},
	month = sep,
	year = {2017},
	note = {arXiv: 1709.09356}
}

@article{morrison_spike-timing-dependent_2007,
	title = {Spike-timing-dependent plasticity in balanced random networks},
	volume = {19},
	issn = {0899-7667},
	doi = {10.1162/neco.2007.19.6.1437},
	language = {english},
	number = {6},
	journal = {Neural Computation},
	author = {Morrison, Abigail and Aertsen, Ad and Diesmann, Markus},
	month = jun,
	year = {2007},
	pmid = {17444756},
	pages = {1437--1467}
}

@article{morrison_phenomenological_2008,
	title = {Phenomenological models of synaptic plasticity based on spike timing},
	volume = {98},
	issn = {0340-1200},
	url = {https://www.ncbi.nlm.nih.gov/pmc/articles/PMC2799003/},
	doi = {10.1007/s00422-008-0233-1},
	number = {6},
	urldate = {2019-09-12},
	journal = {Biological Cybernetics},
	author = {Morrison, Abigail and Diesmann, Markus and Gerstner, Wulfram},
	month = jun,
	year = {2008},
	pmid = {18491160},
	pmcid = {PMC2799003},
	pages = {459--478}
}

@inproceedings{reynaud-bouret_inference_2013,
	title = {Inference of functional connectivity in {Neurosciences} via {Hawkes} processes},
	doi = {10.1109/GlobalSIP.2013.6736879},
	booktitle = {2013 {IEEE} {Global} {Conference} on {Signal} and {Information} {Processing}},
	author = {Reynaud-Bouret, P. and Rivoirard, V. and Tuleau-Malot, C.},
	month = dec,
	year = {2013},
	pages = {317--320}
}

@Article{robert_stochastic_2020,
title = 	 {Averaging Principles for {M}arkovian Models of Plasticity},
author = {Robert, Philippe and Vignoud, Ga{\"e}tan},
month = jun,
year = 2021,
volume = {183},
number = {3},
pages={47--90},
journal = 	 {Journal of Statistical Physics},
url = {https://doi.org/10.1007/s10955-021-02785-3}
}

@Article{robert_stochastic_2020_1,
author = {Robert, Philippe and Vignoud, Ga{\"e}tan},
title = {Stochastic Models of Neural Plasticity},
journal =        {SIAM Journal on Applied Mathematics},
month = sep,
volume = 81,
number = {5},
pages={1821--1846},
year =   2021,
url = {https://doi.org/10.1137/20M138288X},
}

@Article{robert_stochastic_2020_2,
author = {Robert, Philippe and Vignoud, Ga{\"e}tan},
title = {Stochastic {Models} of {Neural} {Plasticity}: {A} {Scaling} {Approach}},
journal =        {SIAM Journal on Applied Mathematics},
month = aug,
year =   2021,
note =   { To Appear.  Arxiv preprint \href{https://arxiv.org/abs/2106.04845}{PDF}},
url = {https://arxiv.org/abs/2106.04845}}

@article{roberts_computational_1999,
	title = {Computational {Consequences} of {Temporally} {Asymmetric} {Learning} {Rules}: {I}. {Differential} {Hebbian} {Learning}},
	volume = {7},
	issn = {1573-6873},
	shorttitle = {Computational {Consequences} of {Temporally} {Asymmetric} {Learning} {Rules}},
	url = {https://doi.org/10.1023/A:1008910918445},
	doi = {10.1023/A:1008910918445},
	language = {english},
	number = {3},
	urldate = {2019-11-08},
	journal = {Journal of Computational Neuroscience},
	author = {Roberts, Patrick D.},
	month = nov,
	year = {1999},
	pages = {235--246}
}

@article{PhysRevE.62.4077,
  title = {Dynamics of temporal learning rules},
  author = {Roberts, Patrick D.},
  journal = {Phys. Rev. E},
  volume = {62},
  issue = {3},
  pages = {4077--4082},
  numpages = {0},
  year = {2000},
  month = {Sep},
  publisher = {American Physical Society},
  doi = {10.1103/PhysRevE.62.4077},
  url = {https://link.aps.org/doi/10.1103/PhysRevE.62.4077}
}

@article{rubin_equilibrium_2001,
	title = {Equilibrium properties of temporally asymmetric {Hebbian} plasticity},
	volume = {86},
	issn = {0031-9007},
	doi = {10.1103/PhysRevLett.86.364},
	language = {english},
	number = {2},
	journal = {Physical Review Letters},
	author = {Rubin, J. and Lee, D. D. and Sompolinsky, H.},
	month = jan,
	year = {2001},
	pmid = {11177832},
	pages = {364--367}
}

@article{haas_spike-timing-dependent_2006,
	title = {Spike-timing-dependent plasticity of inhibitory synapses in the entorhinal cortex},
	volume = {96},
	issn = {0022-3077},
	doi = {10.1152/jn.00551.2006},
	language = {english},
	number = {6},
	journal = {Journal of Neurophysiology},
	author = {Haas, Julie S. and Nowotny, Thomas and Abarbanel, H. D. I.},
	month = dec,
	year = {2006},
	pmid = {16928795},
	pages = {3305--3313},
}

@article{citri_synaptic_2008,
	title = {Synaptic plasticity: multiple forms, functions, and mechanisms},
	volume = {33},
	issn = {0893-133X},
	shorttitle = {Synaptic plasticity},
	doi = {10.1038/sj.npp.1301559},
	language = {english},
	number = {1},
	journal = {Neuropsychopharmacology: Official Publication of the American College of Neuropsychopharmacology},
	author = {Citri, Ami and Malenka, Robert C.},
	month = jan,
	year = {2008},
	pmid = {17728696},
	pages = {18--41},
}

@article{takeuchi_synaptic_2014,
	title = {The synaptic plasticity and memory hypothesis: encoding, storage and persistence},
	volume = {369},
	issn = {1471-2970},
	shorttitle = {The synaptic plasticity and memory hypothesis},
	doi = {10.1098/rstb.2013.0288},
	language = {english},
	number = {1633},
	journal = {Philosophical Transactions of the Royal Society of London. Series B, Biological Sciences},
	author = {Takeuchi, Tomonori and Duszkiewicz, Adrian J. and Morris, Richard G. M.},
	month = jan,
	year = {2014},
	pmid = {24298167},
	pmcid = {PMC3843897},
	pages = {20130288}
}

@article{van_rossum_stable_2000,
	title = {Stable {Hebbian} learning from spike timing-dependent plasticity},
	volume = {20},
	issn = {1529-2401},
	language = {english},
	number = {23},
	journal = {The Journal of Neuroscience: The Official Journal of the Society for Neuroscience},
	author = {van Rossum, M. C. and Bi, G. Q. and Turrigiano, G. G.},
	month = dec,
	year = {2000},
	pmid = {11102489},
	pages = {8812--8821}
}

\appendix
\section{Computer methods}
\label{app:computer}
For each set of parameters, we have run several simulations, with different initial weight values uniformly taken
in $[0, w_{\text{\rm max}}]$.
We have tested the dynamics of the synaptic weight for the different pairing schemes defined before for
a wide range of parameters.
Simulations have been done using Python 3.X for the simple network of a pre-synaptic and a post-synaptic neuron.
We used a discrete Euler scheme for the dynamics of the membrane potential $X$ and the
plasticity variables $Z_{1}$ and $Z_{2}$.
Whenever the synaptic weight was either $0$ or a maximal value $w_{\text{\rm max}}$ the dynamics was stopped
and the synaptic weight state recorded.

We also plot the temporal dynamics for specific values of $B_1$ and $B_2$, typically used $P{=}50$ simulations for each scaling $\eps$.

\section{Pair-based STDP with different pairing schemes}
\label{app:pi}
\subsubsection*{All-to-all Model}

The \emph{all-to-all} pair-based model supposes that all pairs of spikes are taken into account in the synaptic
plasticity rule.
The synaptic weight is updated at each post-synaptic spike occurring at time
$t_{\rm post}$, by taking into account all   pre-synaptic spikes before time $t_{\rm post}$:
\[
  \Delta W (t_{\rm post}) = B_1\sum_{t_{{\rm pre},n}{<}t_{\rm post}}e^{-\gamma_1(t_{\rm post}{-}t_{{\rm pre}, n})}
  = Z_1^{\rm PA}(t_{\rm post})
\]
and,
\[
  \Delta W (t_{\rm pre}) = B_2\sum_{t_{{\rm post}, n}{<}t_{\rm pre}}e^{-\gamma_2(t_{\rm pre}{-}t_{{\rm post}, n})}
  = Z_2^{\rm PA}(t_{\rm pre})
\]
The processes $(Z_i^{\rm PA}(t))$, $i{=}1$, $2$ can be expressed as solutions of the stochastic differential equations,
\begin{equation}\label{SDEAA}
\begin{cases}
\diff Z_1^{\rm PA}(t) \displaystyle =   {-}\gamma_1 Z_1^{\rm PA}(t) \diff t+B_1\mathcal{N}_{\lambda}(\diff t),\\
\diff Z_2^{\rm PA}(t) \displaystyle =   {-}\gamma_2 Z_2^{\rm PA}(t) \diff t+B_2\mathcal{N}_{\beta,X}(\diff t),
\end{cases}
\end{equation}
they are two shot-noise processes, see~\cite{gilbert_amplitude_1960,robert_stochastic_2020_1}.

The synaptic weight updates correspond to the evaluation of $(Z_1^{\rm PA}(t))$ at jumps of the point process
${\cal N}_{\beta, X}$ for post-synaptic activity, and similarly for $(Z_2^{\rm PA}(t))$ with
${\cal N}_{\lambda}$,
\begin{align*}
    \diff W^{\rm PA}(t) &= \sum_{t_{{\rm pre}, n}}Z_2^{\rm PA}(t_{{\rm pre}, n}-)\delta_{t_{{\rm pre}, n}}\\
                        &\hspace{1cm }+ \sum_{t_{\rm post}, n}Z_1^{\rm PA}(t_{{\rm post}, n}-)\delta_{t_{\rm post}, n},
\end{align*}
or, equivalently,
\[
    \diff W^{\rm PA}(t) = Z_2^{\rm PA}(t-){\cal N}_{\lambda}(\diff t) + Z_1^{\rm PA}(t-){\cal N}_{\beta, X}(\diff t).
\]
The notation $U(t{-})$ is for the left limit of the function $(U(t))$ at $t$.
A simple example of the dynamics of the all-to-all pair-based model is depicted in Figure~\ref{fig:markovpairbased} (A)
with interacting pairs of spikes.

\subsubsection*{Nearest-neighbor symmetric model}

In the {\em nearest neighbor symmetric} model, whenever one neuron spikes, the synaptic weight is updated by only
taking into account the last spike of the other neuron, as can be seen in Figure~\ref{figsub:pair}.
If the pre-synaptic neuron fires at time $t_{\rm pre}$, the contribution to the plasticity kernel is
$\Phi(t_{\rm pre}{-}t_{\rm post})$ , where $t_{\rm post}$ is the last post-synaptic spike before $t_{\rm pre}$
and similarly for post-synaptic spikes.

The nearest neighbor symmetric rule leads to,
\begin{equation}\label{SDEPNS}
\begin{cases}
\diff Z_1^{\rm PNS}(t) \displaystyle &=
        {-}\gamma_1 Z_1^{\rm PNS}(t) \diff t\\
        &\hspace{-10mm}+(B_1 - Z_1^{\rm PNS}(t-))\mathcal{N}_{\lambda}(\diff t),\\
\diff Z_2^{\rm PNS}(t) \displaystyle &=
        {-}\gamma_2 Z_2^{\rm PNS}(t) \diff t\\
        &\hspace{-10mm}+(B_2 - Z_2^{\rm PNS}(t-)))\mathcal{N}_{\beta,X}(\diff t).
\end{cases}
\end{equation}
At each pre-synaptic spike, $(Z_1^{\rm PNS}(t))$, resp. $(Z_2^{\rm PNS}(t))$,  is reset to $B_1$, resp. $B_2$. See Figure~\ref{fig:markovpairbased} (B).

\subsubsection*{Nearest-neighbor symmetric reduced model}

Finally, for the {\em nearest neighbor symmetric reduced} scheme, only consecutive pairs of spikes are used to update the
synaptic weight. The synaptic weight is updated at pre-synaptic spike time $t_{\rm pre}$ only if there are no pre-synaptic spikes
 since the last post-synaptic spike. And similarly for post-synaptic spike times.
See Figure~\ref{figsub:pair}~(bottom right).

This rule leads to $(Z_i^{\rm PNR}(t))$, $i{=}1$, $2$, solutions of
\begin{equation}\label{SDEPNR}
\begin{cases}
\diff Z_1^{\rm PNR}(t) \displaystyle &=
        {-}\gamma_1 Z_1^{\rm PNR}(t) \diff t\\
        &\hspace{-10mm}+(B_1 - Z_1^{\rm PNR}(t-))\mathcal{N}_{\lambda}(\diff t)\\
        &\hspace{-10mm}- Z_1^{\rm PNR}(t-))\mathcal{N}_{\beta,X}(\diff t),\\
\diff Z_2^{\rm PNR}(t) \displaystyle &=
        {-}\gamma_2 Z_2^{\rm PNR}(t) \diff t\\
        &\hspace{-10mm}+(B_2 - Z_2^{\rm PNR}(t-)))\mathcal{N}_{\beta,X}(\diff t)\\
        &\hspace{-10mm}- Z_2^{\rm PNR}(t-)\mathcal{N}_{\lambda}(\diff t).
\end{cases}
\end{equation}
See Figure~\ref{fig:markovpairbased} (C).

\section{Slow-fast approximations, averaging principles}
\label{app:slowfast}
\label{secsec:theorem}

We have the scaled system, for $\eps{>}0$,
\begin{equation}
\label{eq:pbsystemscaled}
\begin{cases}\diff X_\eps(t)  &= {-}1/\eps X_\eps(t)\diff t+W_\eps(t{-})\mathcal{N}_{\lambda/\eps}(\diff t),\\
    \diff Z_{1,\eps}(t) &=   {-}\gamma_1 Z_{1,\eps}(t) \diff t/\eps \\
        &\hspace{2mm}+(B_1-K_{1,1}Z_{1,\eps}(t{-}))\mathcal{N}_{\lambda/\eps}(\diff t)\\
        &\hspace{2mm}-K_{1,2}Z_{1,\eps}(t{-})\mathcal{N}_{\beta/\eps,X_\eps}(\diff t),\\
    \diff Z_{2,\eps}(t) &=   {-}\gamma_2 Z_{2,\eps}(t) \diff t/\eps \\
    &\hspace{2mm}+ (B_2{-}K_{2,2}Z_{2,\eps}(t{-}))\mathcal{N}_{\beta/\eps,X_\eps}(\diff t)\\
    &\hspace{2mm}-K_{2,1}Z_{2,\eps}(t{-})\mathcal{N}_{\lambda/\eps}(\diff t),\\
    \diff W_\eps(t) &=  Z_{1,\eps}(t{-})\eps\mathcal{N}_{\beta/\eps,X_\eps}(\diff t)\\
    &\hspace{2mm}+Z_{2,\eps}(t{-})\eps\mathcal{N}_{\lambda/\eps}(\diff t)
\end{cases}
\end{equation}
where  $\gamma_{1},\gamma_{2}{>}0$, $B_{1},B_{2}{\in}\R$,  ${\mathbf K}{=}(K_{ij},i,j{\in}\{1,2\}){\in}\{0,1\}^4$.

Approximations of $(W_\eps(t))$ solution of~\eqref{eq:pbsystem} when $\eps$ is small are discussed and investigated with ad-hoc methods.
The corresponding scaling results, known as separation of timescales, are routinely used in approximations in
mathematical models of computational neuroscience, for example~\cite{kempter_hebbian_1999}.

We first need to define the processes $(X^w(t),Z_1^w(t),Z_2^w(t))$ which follow
the fast processes dynamics with a constant synaptic weight $w$ and prove that a unique invariant distribution
exists for the associated dynamics.
This is the purpose of Proposition~\ref{EFPP}.

\begin{proposition}[Equilibrium of Fast Processes]
\label{EFPP}
For ${\mathbf K}{=}(K_{ij},i,j{\in}\{1,2\}){\in}\{0,1\}^4$, $\gamma_{1},\gamma_{2}{>}0$, $B_{1},B_{2}{\in}\R$, and each $w{\ge}0$, the Markov process $(X^w(t),Z_1^w(t),Z_2^w(t))$ solution of~\eqref{def:pairstati} has a unique stationary distribution $\Pi_w^{\mathbf K}$ on $\R_+{\times}\R^2$.
\end{proposition}
\begin{proof}
See Proposition~25 of~\cite{robert_stochastic_2020}.
\end{proof}

\begin{theorem}[Averaging Principle]
\label{th:PairStoch}
There exists $S_0{\in}(0,{+}\infty]$ such that, when $\eps$ goes to $0$, 
the process $(W_{\eps}(t), t{<}S_0)$ is converging in distribution to $(w(t),t{<}S_0)$,  solution of the equation
\[
\frac{\diff w}{\diff t}(t) = \E_{\Pi^{\mathbf K}_{w(t)}}\left[\lambda Z_{2} {+} \beta(X) Z_{1} \right],
\]
where $\Pi^{K}_{w}$ is defined in Proposition~\ref{EFPP}.
\end{theorem}
\begin{proof}
 See~\cite{robert_stochastic_2020_2} and~\cite{robert_stochastic_2020}.
\end{proof}

\section{Comparison to classical computational models}
\label{app:comp}
In this section, we compare averaging principles for STDP rules leading to Relation~\eqref{AsymW} with the results of~\cite{kempter_hebbian_1999}
in the all-to-all pair-based scheme.

The asymptotic behavior of the synaptic weight dynamics, Relation~(4) of~\cite{kempter_hebbian_1999}, is a consequence of
a similar slow-fast argument,
\begin{equation}\label{eqKemp}
 \frac{\diff \widetilde{w}}{\diff t}(t)=\int_{-\infty}^{+\infty}\widetilde{\Phi}(s)\widetilde{\mu}(s,t) \diff s,
\end{equation}
where,
\begin{itemize}
\item $\widetilde{\Phi}(s)$ represents the STDP curve;
\item $\widetilde{\mu}(s,t){=}\overline{{<} S^{1}(t{+}s)S^{2}(t){>}}$, the correlation between the spike trains.
\end{itemize}

The quantity $\overline{\left<{\cdots}\right>}$ is defined in terms of {\em temporal and ensemble averages},  ${<}{\cdots}{>}$ is the {ensemble average} and $\overline{\cdots}$ the
{temporal average} over the spike trains.

In our setting, Theorem~\ref{th:PairStoch} gives the following equation,
\[
\frac{\diff w}{\diff t}(t) = \E_{\Pi^{\text{\rm PA}}_{w(t)}}\left[\lambda Z_{2} {+} \beta(X) Z_{1} \right],
\]
with
\[
    \Phi(t)\steq{def}B_{1}\exp(-\gamma_{1}t)\ind{t>0} + B_{2}\exp(\gamma_{2}t)\ind{t<0}.
 \]

We have, using simple calculus,
\[
\lambda \E_{\Pi^{\text{\rm PA}}_{w(t)}}\left[Z_{2}\right]
= \int_{-\infty}^0 B_2 \exp(\gamma_2 \tau)\lambda \E_{\Pi^{\text{\rm PA}}_{w(t)}}\left[\beta(x)\right]\diff \tau
\]
We denote by $\Pi^{\text{\rm PA}}_{2{\mapsto}1, t}(\tau)$ the probability of having a post-pre pairing with delay $\tau$ at time $t$.
For the post-pre pairing, we can consider that $\Pi^{\text{\rm PA}}_{2{\mapsto}1, t}(\tau)$ does not depend on $\tau$
and that it is just equal to the product of both rates, i.e there is no
causality, and
\[
    \Pi^{\text{\rm PA}}_{2{\mapsto}1, t}(\tau) = \lambda \E_{\Pi^{\text{\rm PA}}_{w(t)}}\left[\beta(x)\right].
\]
We  easily conclude that,
\[
\lambda \E_{\Pi^{\text{\rm PA}}_{w(t)}}\left[Z_{2}\right]
= \int_{-\infty}^0 \Phi(\tau) \Pi^{\text{\rm PA}}_{2{\mapsto}1, t}(\tau) \diff \tau,
\]
with $\Pi^{\text{\rm PA}}_{2{\mapsto}1, t}(\tau) \approx \overline{{<} S^{1}(t{+}\tau)S^{2}(t){>}}$.

Similarly, we have
\[
    \E_{\Pi^{\text{\rm PA}}_{w(t)}}\left[\beta(X) Z_{1}\right] =
        \E_{\Pi^{\text{\rm PA}}_{w(t)}}\left[\sum_{t_{\text{\rm pre}}{<}t_{\text{\rm post}}}
         B_1 \exp(-\gamma_1 (t_{\text{\rm post}} - t_{\text{\rm pre}})) \beta(X)\right]
\]
We denote by $\Pi^{\text{\rm PA}}_{1{\mapsto}2, t}(\tau)$ the probability of having a pre-post pairing with delay $\tau$ at time $t$.
For the pre-post pairing, this quantity depends on $\tau$ because spikes of the pre-synaptic neuron
influence the spiking of the post-synaptic one, so we have, by using the fact that $\Pi^{\text{\rm PA}}$ is the
invariant distribution,
\[
    \E_{\Pi^{\text{\rm PA}}_{w(t)}}\left[\sum_{t_{\text{\rm pre}}{<}t_{\text{\rm post}}}
         B_1 \exp(-\gamma_1 (t_{\text{\rm post}} - t_{\text{\rm pre}})) \beta(X)\right] =\\
         \int_{0}^{+\infty} B_1 \exp(-\gamma_1 \tau) \Pi^{\text{\rm PA}}_{1{\mapsto}2, t}(\tau) \diff \tau.
\]
See SM2 of~\cite{robert_stochastic_2020_2}, hence
\[
 \E_{\Pi^{\text{\rm PA}}_{w(t)}}\left[\beta(X) Z_{1}\right]
= \int_{0}^{+\infty} \Phi(\tau) \Pi^{\text{\rm PA}}_{1{\mapsto}2, t}(\tau) \diff \tau,
\]
with    $\Pi^{\text{\rm PA}}_{1{\mapsto}2, t}(\tau) \approx \overline{{<} S^{1}(t)S^{2}(t{+}\tau){>}}$.
This shows the equivalence between~\cite{kempter_hebbian_1999} and our result for the all-to-all pair-based STDP rules.

\section{Proofs}

\subsection{All-to-all STDP at an excitatory synapse}\ \\
\label{app:proofPAEXC}

We prove that,
\[
    \E_{\Pi^{\textup{PA}}_{w(t)}}\left[\lambda Z_{2} {+} \beta(X) Z_{1} \right] = A^{\text{\rm PA}}_0 {+} A^{\text{\rm PA}}_{1}w = A^{\text{\rm PA}}_1\left(w{-}w^{\text{\rm PA}}\right).
\]
where,
\[
    A^{\text{\rm PA}}_{0} = \nu\lambda\left(\frac{B_1}{\gamma_1}{+}\frac{B_2}{\gamma_2}\right),\,
    A^{\text{\rm PA}}_{1} = \beta\lambda^2\left(\frac{B_1}{\gamma_1}{+}\frac{B_2}{\gamma_2}{+}\frac{B_1}{\lambda(1{+}\gamma_1)}\right)
\]

\begin{proof}

First, it is easy to show that,
\[
    \E\left[Z^{\textup{PA},w}_1\right] = \lambda \frac{B_1}{\gamma_1}, \text{ and } \E\left[Z^{\textup{PA},w}_2\right]
    = \nu \frac{B_2}{\gamma_2}{+}\beta \lambda \frac{B_2}{\gamma_2} w
\]

Moreover, denoting $(Y^{w}(t)){=}(X^w(t)Z^{\textup{PA},w}_1(t))$, we get
\[
    \diff Y^{w}(t) = {-}(1{+}\gamma_1) Y^w(t)\diff t
        +\left(\rule{0mm}{4mm}wZ^{\textup{PA},w}_1(t{-}){+}B_1 X^w(t{-})+wB_1\right){\cal N}_\lambda(\diff t),
\]
by integrating this ODE on $[0,t]$ and taking the expected value, we obtain

\[
 \E\left[X^wZ^{\textup{PA},w}_1\right]=\frac{
 \lambda w \E\left[Z^{\textup{PA},w}_1\right]{+} \lambda B_1\E\left[X^w\right]{+}\lambda w B_1}{1{+}\gamma_1}
 = \left(\frac{\lambda^2}{\gamma_1}{+}\frac{\lambda}{1{+}\gamma_1}\right) B_1 w .
\]

\end{proof}

\subsection{Nearest neighbor symmetric STDP at an excitatory synapse}
\label{app:proofPNSEXC}

\subsubsection{Estimation of $f^{\text{\rm PNS}}$}
\label{app:proofPNSEXC_f}
~~\\

\[
    \E_{\Pi^{\textup{PNS}}_{w(t)}}\left[\lambda Z_{2} {+} \beta(X) Z_{1} \right] = A^{\text{\rm PNS}}_0{+}A^{\text{\rm PNS}}_1w{+}A^{\text{\rm PNS}}_2 h^{\text{\rm PNS}}(w)
\]
with,
\[
    A^{\text{\rm PNS}}_0{=} \frac{\nu\lambda}{\lambda{+}\gamma_1} B_1{+}\frac{\nu\lambda}{\nu{+}\gamma_2} B_2,\,
    A^{\text{\rm PNS}}_1{=} \lambda \beta \frac{1{+}\lambda}{1{+}\lambda{+}\gamma_1} B_1,\ A^{\text{\rm PNS}}_2{=}\lambda B_2,
\]
and,

\begin{multline*}
    h^{\text{\rm PNS}}(w)=\gamma_2 \int_{\R_+}e^{-\gamma_2\tau}\left(1- \exp\left(\rule{0mm}{4mm}{-}\nu \tau
    {-}\lambda \hspace{-1mm}\int_{0}^{\tau} \hspace{-2mm} \left(1{-}\exp\left({-}\beta w \left(1{-}e^{s-\tau}\right)\right)\right)\diff s\right.\right.\\
    \displaystyle\hspace{3cm}\left.\left.\rule{0mm}{4mm}{-}\lambda \int_{-\infty}^{0} \hspace{-2mm} \left(1{-}\exp\left({-}\beta w \left(1{-}e^{-\tau}\right)e^s\right)\right)\diff s\right)\right) \diff \tau{-}\frac{\nu}{\nu{+}\gamma_2}.
\end{multline*}

\begin{proof}
For $w{\ge}0$, we have,
\begin{multline*}
    f^{\text{\rm PNS}}(w) = \nu B_2 \int_{\R_+} \lambda e^{-(\lambda + \gamma_2) \tau}\diff \tau
   {+}\lambda \beta w B_1 \int_{\R_+} (1{+}\lambda)e^{-(1{+}\lambda{+}\gamma_1) \tau}\diff \tau\\
   {-}\lambda \gamma_2 \int_{\R_+}\exp(-\gamma_2 \tau)\left(1- \exp\left(\rule{0mm}{4mm}{-}\nu \tau
    {-}\lambda \hspace{-1mm}\int_{0}^{\tau} \hspace{-2mm} \left(1{-}\exp\left(-\beta w \left(1{-}e^{s-\tau}\right)\right)\right)\diff s\right.\right.\\
    \displaystyle\hspace{1cm}\left.\left.\rule{0mm}{4mm}{-}\lambda \int_{-\infty}^{0} \hspace{-2mm} \left(1{-}\exp\left(-\beta w \left(1{-}e^{-\tau}\right)e^s\right)\right)\diff s\right)\right) \diff \tau.
\end{multline*}

Stochastic calculus gives, for $\xi{\ge}0$,
\begin{multline*}
  {-}\ln \E\left[e^{{-}\xi{\cal N}_{\beta,X^w_\infty}((0,a))}\right]
  {=}\nu a\left(\!1{-}e^{-\xi}\!\right)
{+}\lambda\!\! \int_{0}^{a} \hspace{-2mm} \left(1{-}\exp\left(-\beta w \left(1{-}e^{{-}\xi}\right)\!\!\left(1{-}e^{s-a}\right)\right)\right)\!\diff s\\
{+}\lambda \int_{-\infty}^{0} \hspace{-2mm} \left(1{-}\exp\left(-\beta w \left(1{-}e^{-\xi}\right)\left(1{-}e^{-a}\right)e^s\right)\right)\diff s.
\end{multline*}
By letting $\xi$ go to infinity, we have obtained the desired expression. The proposition is proved.

\end{proof}

\subsubsection{Dynamics of $w$}\ \\
\label{app:proofPNSEXC_dynamics}

We start with some calculations for $h^{\text{\rm PNS}}$ of Section~\ref{secsec:excPN}.
\begin{lemma}
$h^{\text{\rm PNS}}$ is a convex function, and,
\[
     h^{\text{\rm PNS}}(0) = 0,\, h^{\text{\rm PNS}}({+}\infty) = 1{-}\frac{\nu}{\nu{+}\gamma_2},\,
     {h^{\text{\rm PNS}}}^{\prime}(0)=\frac{\lambda\beta\gamma_2}{(\nu{+}\gamma_2)^2}
\]
\end{lemma}

\begin{proof}
We compute,

\begin{multline*}
    {h^{\text{\rm PNS}}}^{\prime}(w)=\lambda\beta\gamma_2\int_{\R_+}e^{-(\nu+\gamma_2)\tau}\left(\hspace{-1mm}\int_{0}^{\tau} \hspace{-2mm}\left(1{-}e^{s-\tau}\right)\exp\left(-\beta w \left(1{-}e^{s-\tau}\right)\right)\diff s\right.\\
    \left.{+}\int_{-\infty}^{0} \hspace{-2mm}\left(1{-}e^{-\tau}\right)e^s \exp\left(-\beta w \left(1{-}e^{-\tau}\right)e^s\right)\diff s \right)\\
    \exp\left(\rule{0mm}{4mm}{-}\lambda \hspace{-1mm}\int_{0}^{\tau} \hspace{-2mm} \left(1{-}\exp\left(-\beta w \left(1{-}e^{s-\tau}\right)\right)\right)\diff s
    {-}\lambda \int_{-\infty}^{0} \hspace{-2mm} \left(1{-}\exp\left(-\beta w \left(1{-}e^{-\tau}\right)e^s\right)\right)\diff s\right) \diff \tau.
\end{multline*}

We have $h'(w)$ is an increasing function in $w$, so $h(w)$ is convex.
\end{proof}
The system

\[
    \frac{\diff w}{\diff t}(t) = A^{\text{\rm PNS}}_0{+}A^{\text{\rm PNS}}_1w{+}A^{\text{\rm PNS}}_2 h^{\text{\rm PNS}}(w)
\]
has the following dynamics.

\begin{table*}[h]
\renewcommand*{\arraystretch}{2}
\centering
\begin{tabular}{ ccccc }
 $\nu$ & LTD & LTP & STABLE FP & UNSTABLE FP \\  \hline
 $0$ & $A^{\text{\rm PNS}}_3{<}0$&  $A^{\text{\rm PNS}}_3{>}0$ & $A^{\text{\rm PNS}}_3{<}0$ & $A^{\text{\rm PNS}}_3{>}0$\\
  & $A^{\text{\rm PNS}}_1{<}0$ &  $A^{\text{\rm PNS}}_1{>}0$ & $A^{\text{\rm PNS}}_1{>}0$ & $A^{\text{\rm PNS}}_1{<}0$ \\ \hline
 $>0$ & $A^{\text{\rm PNS}}_0{<}0$&  $A^{\text{\rm PNS}}_0{>}0$ & $A^{\text{\rm PNS}}_0{<}0$ & $A^{\text{\rm PNS}}_0{>}0$\\
  &$A^{\text{\rm PNS}}_1{<}0$ & $A^{\text{\rm PNS}}_1{>}0$ & $A^{\text{\rm PNS}}_1{>}0$ & $A^{\text{\rm PNS}}_1{<}0$ \\
\end{tabular}
\caption{Bifurcations parameters for the nearest neighbor symmetric scheme}
\label{table:compPNS}
\end{table*}

where
\[
    A^{\text{\rm PNS}}_3 = \lambda \beta\left(\frac{1{+}\lambda}{1{+}\lambda{+}\gamma_1} B_1{+} \frac{\lambda}{\gamma_2} B_2\right) = {f^{\text{\rm PNS}}}^{\prime}(0)_{\nu{=}0}.
\]

\begin{proof}
~~\\
{\it Case $\nu=0$}

We have $f^{\text{\rm PNS}}(0)=0$ and
$\lim_{w\rightarrow+\infty}f^{\text{\rm PNS}}(w)=\text{sign}(A^{\text{\rm PNS}}_1){\times}\infty$.
We need to look then at the sign of ${f^\text{\rm PNS}}^{\prime}(0)=A^\text{\rm PNS}_3$.

If $A^{\text{\rm PNS}}_1$ and $A^{\text{\rm PNS}}_3$ are of the same sign, $f^\text{\rm PNS}_1$ has no positive roots.
Therefore, if $A^{\text{\rm PNS}}_1{>}0$ and $A^{\text{\rm PNS}}_3{>}0$, we have $\lim_{t\rightarrow+\infty}w(t)=+\infty$.
Reciprocally, if $A^{\text{\rm PNS}}_1{>}0$ and $A^{\text{\rm PNS}}_3{<}0$, we have $\lim_{t\rightarrow+\infty}w(t)=0$.

If $A^{\text{\rm PNS}}_1$ and $A^{\text{\rm PNS}}_3$ are not of the same sign, $f^\text{\rm PNS}_1$ has a unique
positive root $w^{\text{\rm PNS}}$.
Then, if $A^{\text{\rm PNS}}_1{<}0$ and $A^{\text{\rm PNS}}_3{>}0$, $w^{\text{\rm PNS}}$ is a stable fixed point and
$A^{\text{\rm PNS}}_1{>}0$ and $A^{\text{\rm PNS}}_3{<}0$, it is an unstable fixed point.

~~\\
{\it Case $\nu>0$}

We have $f^{\text{\rm PNS}}(0)=A^{\text{\rm PNS}}_0$ and
$\lim_{w\rightarrow+\infty}f^{\text{\rm PNS}}(w)=\text{sign}(A^{\text{\rm PNS}}_1){\times}\infty$

Similarly as for $\nu{=}0$, if $A^{\text{\rm PNS}}_0$ and $A^{\text{\rm PNS}}_1$ are not of the same sign, $f^\text{\rm PNS}_1$ has a unique
positive root $w^{\text{\rm PNS}}$, following the convexity of $f^{\text{\rm PNS}}$.
Then, if $A^{\text{\rm PNS}}_0{>}0$ and $A^{\text{\rm PNS}}_1{<}0$, $w^{\text{\rm PNS}}$ is a stable fixed point and
$A^{\text{\rm PNS}}_0{<}0$ and $A^{\text{\rm PNS}}_1{>}0$, it is an unstable fixed point.

It is slightly more complex for the other cases.
We will focus on the case, $A^{\text{\rm PNS}}_0{>}0$ and $A^{\text{\rm PNS}}_1{>}0$.
We have that $f^{\text{\rm PNS}}(0)>0$ and that $\lim_{w\rightarrow+\infty}f^{\text{\rm PNS}}(w)=+\infty$.
As $f^{\text{\rm PNS}}$ is convex, two cases are possible. Either $f^{\text{\rm PNS}}$ has no positive root, and in
that case, it is easy to see that $\lim_{t\rightarrow+\infty}w(t)=+\infty$.
However, it is also possible that $f^{\text{\rm PNS}}$ has two positive roots and in that case it would lead to more
complex dynamics. we just need to look at ${f^{\text{\rm PNS}}}^{\prime}(0)$ and show that it is positive to prove that
this case does not happen.

$A^{\text{\rm PNS}}_0{>}0$ leads to a first inequality,
\[
     B_1 \geq{-}B_2\frac{\lambda{+}\gamma_1}{\nu{+}\gamma_2} \geq 0.
\]

We can then say that,
\begin{align*}
{f^{\text{\rm PNS}}}^{\prime}(0) &= A^{\text{\rm PNS}}_1{+}A^{\text{\rm PNS}}_2\frac{\lambda\beta\gamma_2}{(\nu{+}\gamma_2)^2} = B_1 \lambda \beta \frac{1{+}\lambda}{1{+}\lambda{+}\gamma_1}{+}B_2\frac{\lambda^2\beta\gamma_2}{(\nu{+}\gamma_2)^2}\\
      & \geq{-}B_2\lambda \beta\frac{(\lambda{+}\gamma_1)\frac{1{+}\lambda}{1{+}\lambda{+}\gamma_1}{-}\frac{\lambda\gamma_2}{\nu{+}\gamma_2}}{\nu{+}\gamma_2} = {-}B_2\lambda \beta\frac{\lambda\nu{+}\lambda^2\nu{+}\gamma_1\nu{+}\gamma_1\gamma_2+\lambda\nu\gamma_1}{(\nu{+}\gamma_2)^2(1{+}\lambda{+}\gamma_1)} \geq 0\\
\end{align*}

The same arguments are true for the other case.
\end{proof}

\subsubsection{Approximation for $w$ small}
\label{app:proofPNSEXC_approx}
~~\\

We have the following expansion for $w$ small,
\[
    f^{\text{\rm PNS}}(w) =
    \nu B_1\frac{\lambda}{\lambda{+}\gamma_1}{+}\lambda \beta w B_1\frac{1{+}\lambda}{1{+}\lambda{+}\gamma_1}
   {+}\lambda B_2 \frac{\nu{+} \lambda\beta w}{\gamma_2 {+}\nu{+} \lambda\beta w}{+}o(w).
\]

Leading to the following differential system,
\[
\frac{\diff w}{\diff t}(t) =
    \frac{a^{\text{\rm PNS}}w(t)^2+b^{\text{\rm PNS}}w(t)+c^{\text{\rm PNS}}}{w(t)-w^{\text{\rm PNS}}_{\text{approx}}}+o(w(t)),
\]
where,

\begin{multline*}
    a^{\text{\rm PNS}} = B_1\frac{\lambda^2\beta^2(1{+}\lambda)}{1{+}\lambda{+}\gamma_1},\,
    b^{\text{\rm PNS}} = B_1\frac{\nu\lambda^2\beta}{\lambda{+}\gamma_1}{+}
        B_1\frac{\lambda\beta(1{+}\lambda)(\gamma_2{+}\nu)}{1{+}\lambda{+}\gamma_1}{+}B_2\lambda,\\
    c^{\text{\rm PNS}} = B_1\frac{\nu\lambda(\gamma_2{+}\nu)}{\lambda{+}\gamma_1}{+}B_2\frac{\nu}{\beta} \text{ and }
    w^{\text{\rm PNS}}_{\text{approx}} ={-}\frac{\gamma_2{+}\nu}{\lambda\beta}.
\end{multline*}

\begin{proof}
\begin{multline*}
    f^{\text{\rm PNS}}(w) = \nu B_1\frac{\lambda}{\lambda{+}\gamma_1}{+}\lambda \beta w B_1\frac{1{+}\lambda}{1{+}\lambda{+}\gamma_1}\\
   {+}\lambda B_2{-}\lambda \gamma_2 B_2\int_{\R_+}\exp\left(\rule{0mm}{4mm}{-}\tau\left(\gamma_2
        {+}\nu{+} \lambda\beta w\right) \right) \diff \tau{+}o(w).
\end{multline*}
\end{proof}

Therefore, if
\[
    \Delta^{\text{\rm PNS}} = {b^{\text{\rm PNS}}}^2-4a^{\text{\rm PNS}}c^{\text{\rm PNS}}>0,
\]
we have an analytical expression for the fixed points of the dynamics $w^{\text{\rm PNS}}_{w{\approx}0}$,
\[
    w^{\text{\rm PNS}}_{w{\approx}0} = \frac{-b^{\text{\rm PNS}} +/- \sqrt{\Delta^{\text{\rm PNS}}}}{2a^{\text{\rm PNS}}}.
\]

\subsection{Nearest neighbor symmetric reduced STDP at an excitatory synapse}
\label{app:proofPNREXC}

To study the invariant distribution, we need to use a different formulation of the nearest reduced symmetric rule.

For $w{\ge}0$, we can define $(X^w(t),T^{\textup{PNR}, w}_1, T^{\textup{PNR}, w}_2(t))$, the solution of the SDEs,
\begin{equation}
\label{def:pairstatiPNR}
\begin{cases}
    \diff X^w(t) \displaystyle = {-}X^w(t)\diff t+w\mathcal{N}_{\lambda}(\diff t),\\
    \diff T^{\textup{PNR}, w}_1(t)= \diff t{-}T^{\textup{PNR}, w}_1(t{-})\,\mathcal{N}_{\lambda}(\diff t),\\
    \diff T^{\textup{PNR}, w}_2(t)= \diff t{-}T^{\textup{PNR}, w}_2(t{-})\,\mathcal{N}_{\beta,X^w}(\diff t).
\end{cases}
\end{equation}

and,
\begin{multline*}
     \frac{\diff w}{\diff t}(t) = f^{\text{\rm PNR}}(w) = \E_{\Pi^{\textup{PNR}}_{w(t)}}\left[\lambda Z_{2} {+} \beta(X) Z_{1} \right]\\
        = \E_{\Pi^{\textup{PNR}}_{w(t)}}\left[ \ind{T_1<T_2}B_1 \beta(X)\exp(-\gamma_1 T_1)\right. \\
       {+}\left.\ind{T_2<T_1}B_2\lambda \exp(-\gamma_2 T_2)\right]
\end{multline*}

\begin{figure}[p]\centerfloat
\includegraphics{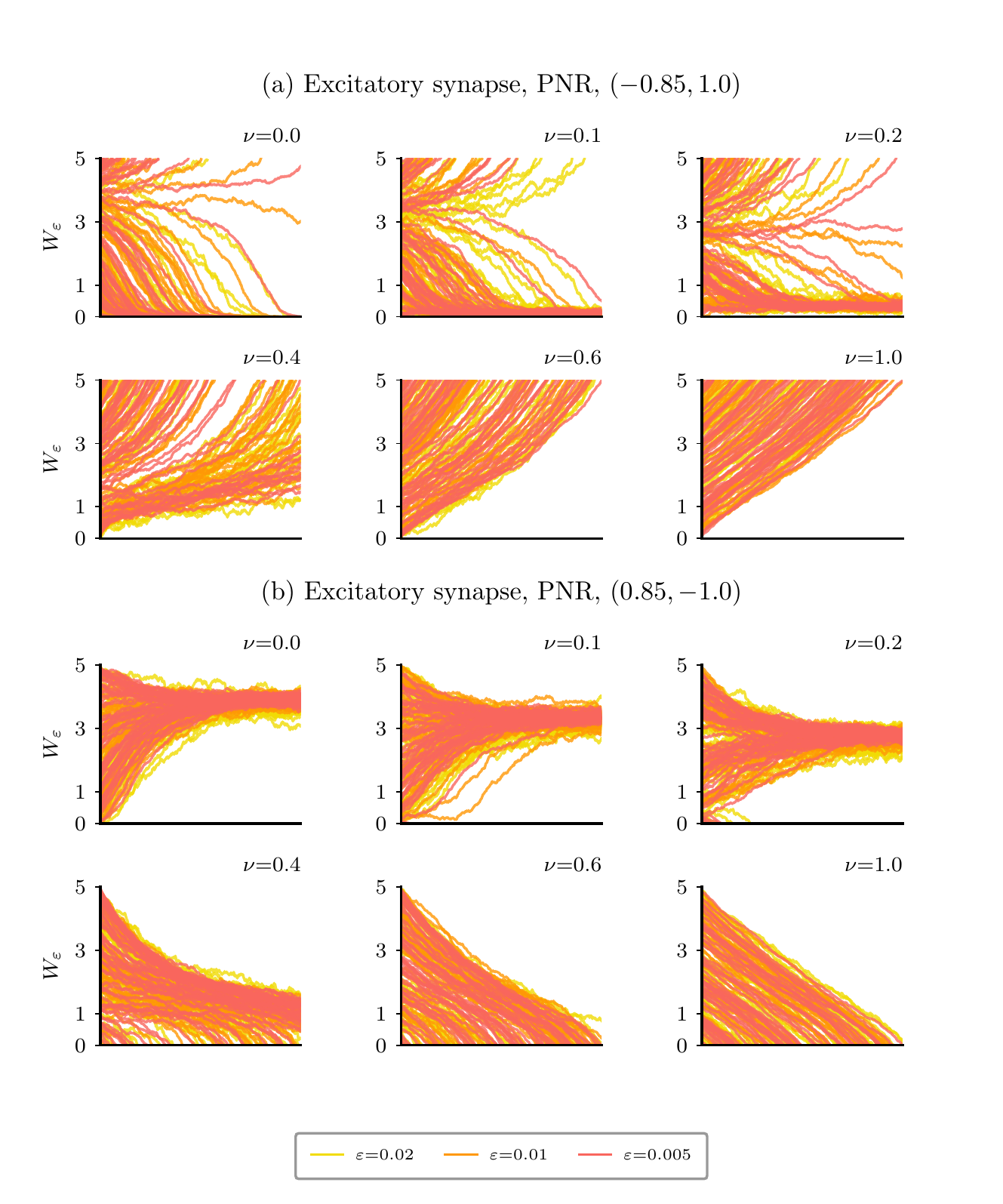}
\caption{Influence of $\nu$ on dynamics with the pair-based nearest-neighbor reduced symmetric scheme}
\label{fig:EXTERNAL}
\end{figure}

\subsection{All-to-all STDP at an inhibitory synapse}
\label{app:proofPAINH}

\begin{definition}
    We define the density of probability $Q(y)$ of the exponential Shot-Noise process $Y$ associated to $\mathcal{N}_{\lambda}$, according to Gilbert and Pollack (1960).
    A general expression of $Q(y)$ can be found in Gilbert and Pollack (1960).
    In our case, we will use,
    \[
    \begin{cases}
        Q(y) &= c(\lambda) y^{\lambda-1}\quad, 0\leq y \leq 1,\\
        \displaystyle\lim_{y\rightarrow+\infty}Q(y) &= 0,
    \end{cases}
    \]
    with,
    \[
        c(\lambda) = \frac{e^{-\gamma_{e} \lambda}}{\Gamma(\lambda)}
    \]
    where, $\gamma_e$ Euler constant and $\Gamma$ Euler function.
\end{definition}

We have the two following limits, for small $w$,
\[
    \frac{\diff w}{\diff t}(t) = f^{\text{\rm PA}}(w) = A^{\text{\rm PA}}_0{-}A^{\text{\rm PA}}_{1}w = -A^{\text{\rm PA}}_1\left(w+w^{\text{\rm PA}}\right).
\]
where,
\[
    A^{\text{\rm PA}}_{0} = \nu\lambda\left(\frac{B_1}{\gamma_1}{+}\frac{B_2}{\gamma_2}\right) \text{ and }
    A^{\text{\rm PA}}_{1} = \beta\lambda^2\left(\frac{B_1}{\gamma_1}{+}\frac{B_2}{\gamma_2}{+}\frac{B_1}{\lambda(1{+}\gamma_1)}\right).
\]

We can compute, when $w{\geq} \nu /\beta$,

\[
\frac{\diff w}{\diff t}(t) = \frac{A^{\text{\rm PAI}}}{w(t)^{\lambda{+}\gamma_1}}\left(1{+}\eta^{\text{\rm PAI}}\left[\frac{w(t)}{w^{\text{\rm PAI}}} \right]^{\gamma_1} \right)
\]
where,
\[
    A^{\text{\rm PAI}} = \left[ \frac{\nu}{\beta} \right]^{\lambda{+}\gamma_1}\frac{c(\lambda) B_1 \nu}{\left( \lambda{+}\gamma_1\right)\left( \lambda{+}\gamma_1{+}1\right)},\,
    w^{\text{\rm PAI}} = \frac{\beta}{\nu} \left(\left| \frac{B_2}{B_1} \right|\frac{\left( \lambda{+}\gamma_1\right)\left( \lambda{+}\gamma_1{+}1\right)}{\gamma_2\left( \lambda{+}1\right)}\right)^{1/\gamma_1},
\]
and,
\[
    \eta^{\text{\rm PAI}} = \left| \frac{B_2}{B_1} \right| \frac{B_1}{B_2}.
\]

\begin{proof}
For $w{\ge}0$, we have to calculate,
\[
    I_1\steq{def}\int (\nu-\beta x)^+z_1\Pi^{\textup{PA}}_w(\diff x,\diff z),
\text{ and }
     I_2\steq{def}\lambda \int  z_2\Pi^{\textup{PA}}_w(\diff x,\diff z).
\]

We have,
\[
    I_2 = \frac{\lambda B_2}{\gamma_2}\int (\nu-\beta x)^+\Pi^{\textup{PA}}_w(\diff x,\diff z)= c(\lambda) \frac{\lambda B_2}{\gamma_2}\int_0^{\frac{\nu}{\beta w}}\left( \nu-\beta w y\right)Q(y)\diff y
\]

We have two cases, if $w{\ll}\nu /\beta$, then,
\[
    I_2 =  \frac{\lambda B_2}{\gamma_2}\left( \nu{-}\beta w \lambda \right).
\]

And, if $w{\geq}\nu /\beta$,
\begin{align*}
    I_2 &= c(\lambda) \frac{\lambda B_2}{\gamma_2}\int_0^{\frac{\nu}{\beta w}}\left( \nu-\beta w y\right)y^{\lambda-1}\diff y \\
    &= c(\lambda) \frac{\lambda B_2}{\gamma_2} \nu \left[ \frac{\nu}{\beta w}\right]^{\lambda} \left( \frac{\nu}{\lambda}{-}\frac{1}{\lambda{+}1} \right)
     = c(\lambda) \frac{B_2}{\gamma_2} \frac{\nu}{\lambda{+}1} \left[ \frac{\nu}{\beta w}\right]^{\lambda}
\end{align*}

Then,
\[
    I_1 =\int \max(0,\nu-\beta x)z_1\Pi^{\textup{PA}}_w(\diff x,\diff z) =c(\lambda) \int_0^{\frac{\nu}{\beta w}} \left(\nu-\beta w y\right)B_1y^{\gamma_1}Q(y)\diff y
\]

We have two cases again, if $w \ll \nu /\beta$, then,
\[
    I_1 = \frac{\lambda B_1}{\gamma_1}\left( \nu{-}\beta w \lambda \right){-} \frac{\lambda B_1}{\gamma_1{+}+1}\beta w
\]

Again, if $w \geq \nu /\beta$,
\begin{align*}
    I_1&=\int \max(0,\nu-\beta x)z_1\Pi^{\textup{PA}}_w(\diff x,\diff z)=c(\lambda) \int_0^{\frac{\nu}{\beta w}} \left(\nu{-}\beta w y\right)B_1y^{\gamma_1}y^{\lambda-1}\diff y\\
    &=\frac{c(\lambda) B_1 \nu}{(\lambda{+}\gamma_1)(\lambda{+}\gamma_1{+}1)} \left[\frac{\nu}{\beta w}\right]^{\lambda{+}\gamma_1}
\end{align*}

\end{proof}

\begin{table}[p]
\renewcommand*{\arraystretch}{2}
\centering
\begin{tabular}{ cccccc }
  & $\nu$ & Sym. LTD & Sym. LTP & Hebbian & Anti-Hebbian \\
 \hline
 \multirow{5}{*}{PA} & \multirow{2}{*}{0} & \multirow{2}{*}{LTD} & \multirow{2}{*}{LTP} & LTD if $A^{\text{\rm PA}}_{0} {<} 0$ & LTD if $A^{\text{\rm PA}}_{0} {<} 0$\\
  \cline{5-6}
  & & & & LTP if $A^{\text{\rm PA}}_{0} {>} 0$ & LTP if $A^{\text{\rm PA}}_{0} {>} 0$\\
  \cline{3-6}
  & \multirow{3}{*}{${>}0$} & \multirow{3}{*}{LTD} & \multirow{3}{*}{LTP} & LTD if $A^{\text{\rm PA}}_{0} {<} 0$ & LTD if $A^{\text{\rm PA}}_{1} {<} 0$\\
  \cline{5-6}
  & & & & LTP if $A^{\text{\rm PA}}_{1} {>} 0$ & LTP if $A^{\text{\rm PA}}_{0} {>}0$\\
  \cline{5-6}
  & & & & UNSTABLE FP if not & STABLE FP if not\\
  \hline
  \multirow{2}{*}{PNS} & \multirow{2}{*}{$\geq 0$} & \multirow{2}{*}{LTD} & \multirow{2}{*}{LTP} & LTP if $A^{\text{\rm PNS}}_{0/3} {>} 0$ & LTD if $A^{\text{\rm PNS}}_{0/3} {<} 0$\\
  \cline{5-6}
  & & & & UNSTABLE FP if not & STABLE FP if not\\
  \hline
  \multirow{6}{*}{PNR*} & \multirow{3}{*}{${=}0$} & \multirow{3}{*}{LTD} & \multirow{3}{*}{LTP} & LTD & LTD \\
  \cline{5-6}
  & & & & LTP & LTP \\
  \cline{5-6}
  & & & & STABLE FP & UNSTABLE FP\\
  \cline{3-6}
  & \multirow{3}{*}{${>}0$} & \multirow{3}{*}{LTD} & \multirow{3}{*}{LTP} & LTD & LTD \\
  \cline{5-6}
  & & & & LTP & LTP \\
  \cline{5-6}
   & & & & MULTIPLE FP & MULTIPLE FP \\
\end{tabular}
\caption{Different pairing schemes lead to diverse dynamics for an excitatory synapse (* with simulations)}
\label{table:compPIexc}
\end{table}

\begin{table}[p]
\renewcommand*{\arraystretch}{2}
\centering
\begin{tabular}{ ccccc }
  & Sym. LTD & Sym. LTP & Hebbian & Anti-Hebbian \\
  \hline
  \multirow{2}{*}{PA} & \multirow{2}{*}{LTD} & \multirow{2}{*}{LTP} & LTD if $A^{\text{\rm PA}}_{0} {<} 0$ & LTP if $A^{\text{\rm PA}}_{0} {>} 0$\\
  \cline{4-5}
  & & & STABLE FP if not & UNSTABLE FP if not\\
  \hline
  \multirow{2}{*}{PNS/PNR*}  & \multirow{2}{*}{LTD} & \multirow{2}{*}{LTP} & LTD  & LTP\\
  \cline{4-5}
  & & & STABLE FP & UNSTABLE FP\\
\end{tabular}
\caption{Different pairing schemes lead to diverse dynamics for an inhibitory synapse}
\label{table:compPIinh}
\end{table}

\begin{figure}[p]
\centerfloat
\includegraphics{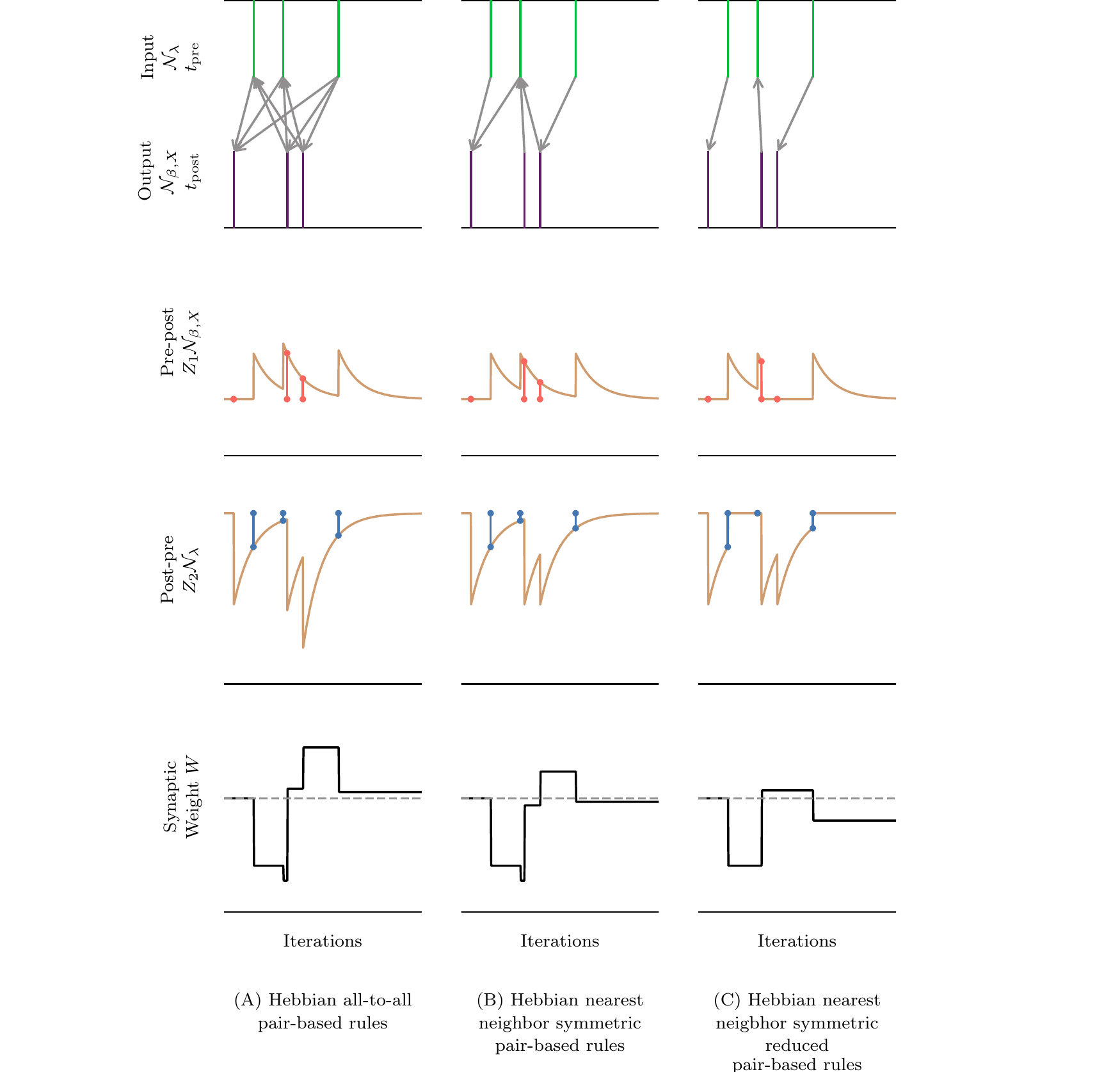}
\caption{Markovian formulation of pair-based models}
\label{fig:markovpairbased}
\end{figure}

\begin{figure}[p]
\centerfloat
\includegraphics{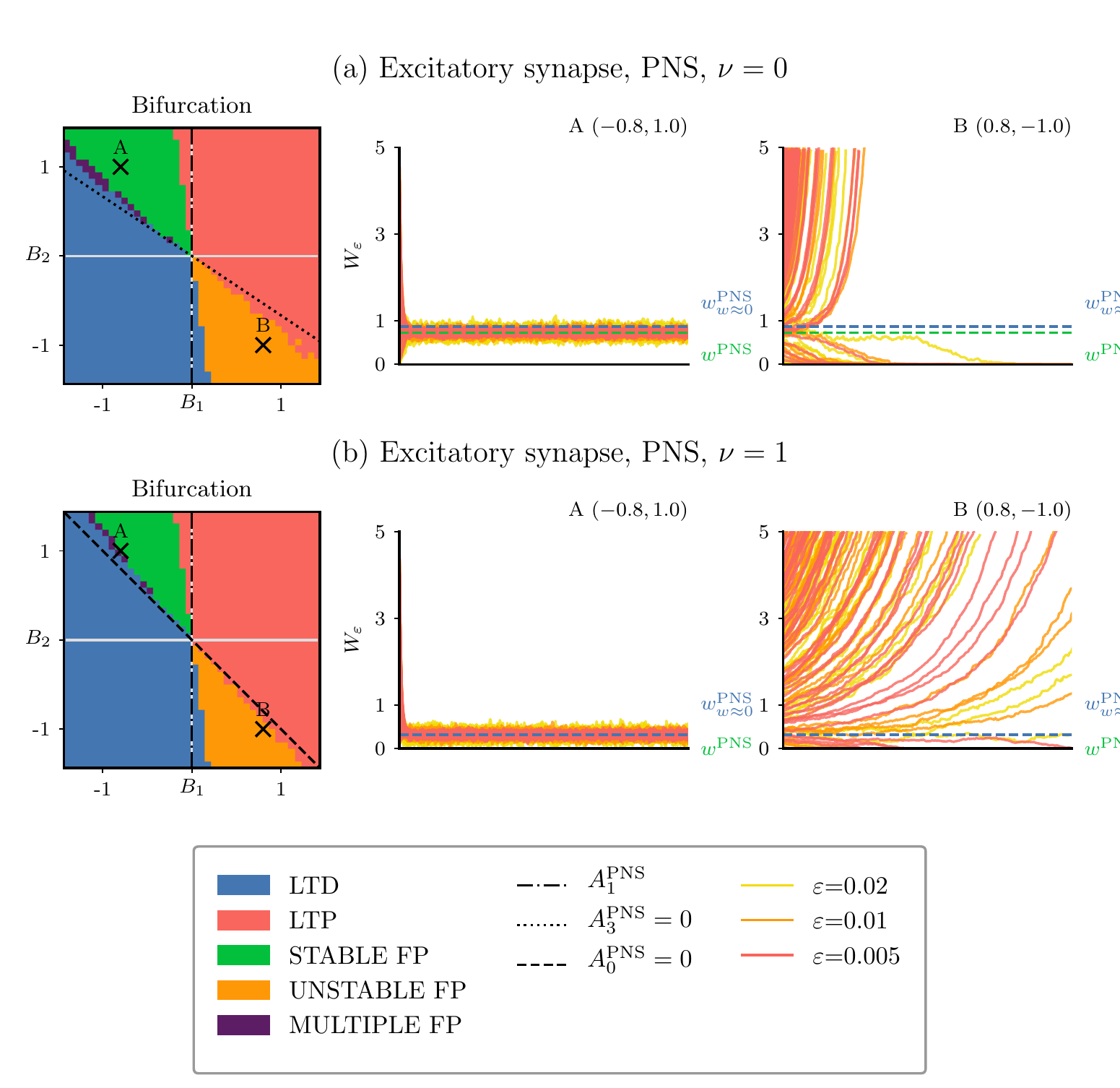}
\caption{Nearest neighbor symetric pair-based STDP for an excitatory synapse}
\label{fig:FigurePIBIF_PNS}
\end{figure}

\end{document}